\documentclass[8pt,journal, twocolumn]{IEEEtran}
\usepackage[T1]{fontenc}
\usepackage[latin9]{inputenc}
\usepackage{geometry}
\usepackage{float}
\usepackage{amsthm}
\usepackage{amsmath}
\usepackage{amssymb}
\usepackage{graphicx}

\makeatletter

\providecommand{\tabularnewline}{\\}

\theoremstyle{plain}
\newtheorem{thm}{\protect\theoremname}
\theoremstyle{plain}
\newtheorem{lem}[thm]{\protect\lemmaname}
\theoremstyle{plain}
\newtheorem{prop}[thm]{\protect\propositionname}

\providecommand{\tabularnewline}{\\}

\theoremstyle{plain}

\usepackage{epstopdf}

\usepackage{units}

\providecommand{\lemmaname}{Lemma}
\providecommand{\propositionname}{Proposition}
\providecommand{\theoremname}{Theorem}

\providecommand{\lemmaname}{Lemma}
\providecommand{\propositionname}{Proposition}
\providecommand{\theoremname}{Theorem}

\makeatother

\providecommand{\lemmaname}{Lemma}
\providecommand{\propositionname}{Proposition}
\providecommand{\theoremname}{Theorem}

\begin{document}

\title{A Convex Relaxation Approach to Higher-Order Statistical Approaches
to Signal Recovery}

\author{Huy-Dung Han$^{(1)}$, Zhi Ding$^{(2)}$, Muhammad Zia$^{(3)}$ %
\thanks{This material is based upon work supported by National Science Foundation under Grants ECCS-1307820, CNS-1443870, CNS-1457060 and by the Vietnam Education Foundation. %
}\\
 {\small $^{(1)}$ Hanoi University of Science and Technology,
          Department of Electronics and Computer Engineering,
					Hanoi, Vietnam\\
$^{(2)}$  University of California,
          Department of Electrical and Computer Engineering,
          Davis, CA, USA \\
$^{(3)}$  Quaid-i-Azam University, 
					Department of Electronics, 
					Islamabad, Pakistan \\
email:\{hdhan,zding,mzia\}@ucdavis.edu
}
}
\maketitle
\begin{abstract}
In this work, we investigate an efficient numerical approach for solving
higher order statistical methods for blind and semi-blind signal
recovery from non-ideal channels. We develop numerical algorithms
based on convex optimization relaxation for minimization of higher
order statistical cost functions. The new formulation through convex
relaxation overcomes the local convergence problem of existing gradient
descent based algorithms and applies to several well-known cost functions
for effective blind signal recovery including blind equalization and
blind source separation in both single-input-single-output (SISO)
and multi-input-multi-output (MIMO) systems. We also propose a fourth order
pilot based cost function that benefits from this approach. The simulation results
demonstrate that our approach is suitable for short-length packet
data transmission using only a few pilot symbols. 
\end{abstract}

\section{Introduction\label{sec:CO-Introduction}}

Blind signal recovery is a well-known problem in signal processing
and communications. With a special goal of recovering unknown input
signals to unknown linear systems based on the system output signals,
this problem typically manifests itself either as blind equalization
or blind source separation. In SISO systems or single-input-multiple-output
(SIMO) systems, the goal of channel equalization is to undo the
inter-symbol interference (ISI). In MIMO systems or in source separation,
the objective is to mitigate both the ISI and the inter-channel
interference (ICI) in order to successfully separate different sources.
The advantage lies in the fact that blind algorithms do not require
allocation of extra bandwidth to training signals whereas semiblind
algorithms can substantially reduce the length of training signals.
System implementations based on blind algorithms have appeared in
downstream cable modem \cite{DOCSIS}, HDTV \cite{Johnson1998,Choi1993,Werner1999}.

The specially designed cost functions for blind signal recovery typically
pose a challenge to the issue of global convergence and convergence
speed. In literature, global convergent algorithms for blind
equalization and source separation do exist for linear programming
equalization \cite{Ding2000} and space-time signal detection \cite{Muhammad2010}.
However, without limiting system input signals to the special QAM
class, most of the blind algorithms are based on non-quadratic and
non-convex costs that utilize high-order statistics of the received
signals. Well-known algorithms of this type include the constant
modulus algorithm (CMA) \cite{Godard1980,Treichler1983}, the Shalvi-Weinstein
algorithm (SWA) \cite{Shalvi1990}, and the minimum entropy deconvolution
(MED) algorithm \cite{Wiggins1978,Donoho1981}. In fact, without modifications,
these stochastic gradient descent (SGD) algorithms typically admit
multiple local minima \cite{Kennedy1990,Li1995} and require large
numbers of data and iterations to converge.

To improve convergence speed of blind channel equalization techniques,
batch algorithms can effectively utilize the statistical information
of the channel output signals. They can significantly shorten the
convergence time \cite{Agee1986,Pickholtz1993,Regalia2002,Chen2004,Zarzoso2008,Satorius1993}.
Unfortunately, local convergence remains a major obstacle, that
requires good initialization strategies. Another approach to mitigate
the local convergence problem is through relaxation by ``lifting''
the receiver parameters to a new parameterization that covers a larger
parameter space over which the cost function is convex \cite{Dogancay1999,Maricic2003}.
Once the convex problem is uniquely solved, the solution is mapped
back to the original restrictive parameter space. For example, the
authors of \cite{Dogancay1999} relaxed the outer-product matrix of
the equalizer parameter space from a rank-1 matrix to an unrestricted
square matrix over which the CMA cost becomes quadratic and can be
solved globally via least squares (LS). From the relaxed solution,
a series of linear operations map the solution back into the right equalizer
coefficient space. In \cite{Maricic2003}, the authors proposed a
similar relaxation with a much more elegant algorithm that modifies
the CMA cost function into a special sum. The modified CMA cost is
first solved over a more restrictive semi-definite positive matrix
space. The resulting convex optimization problem is then solved efficiently
via semidefinite programing (SDP). Iterative mappings must also follow
the SDP solution in order to determine the equalizer parameters \cite{Maricic2003}.

In this work, we study means to improve the convergence of general
blind and semiblind equalization algorithms by generalizing the cost
modification principle presented in \cite{Maricic2003}. We can therefore
develop a more general batch implementation of well-known blind equalization
and source separation algorithms that minimize cost functions involving
the fourth order statistics of system output signals. More specifically,
the new implementation leverages a convex optimization relaxation
that can be applied to CMA, SWA, and MED algorithms. We show that
our convex formulation requires less resource and improves the efficiency
of the convex formulation in \cite{Maricic2003}. We further generalize
the formulation to accommodate the semi-blind algorithms when a small
number of training symbols are available to assist the receiver signal
recovery and separation. Our proposed method overcomes the local convergence
problem which is a drawback of traditional gradient descent implementations.


The rest of the paper is organized as follows. In Section II, we introduce
MIMO system model. Section III discusses batch MIMO blind source recovery
using CMA cost as an example of real-valued fourth order functions.
Section IV presents a convex formulation of the fourth order function
and algorithm to find its global minima. In Section V, we extend formulation
to other blind algorithms based on fourth order cost. A
semiblind algorithm using fourth order function is proposed in Section
VI. In Section VII, we present our simulation results and Section
VIII contains some concluding remarks. In the Appendices, we include
the formulation of converting cross-correlation cost and the training
based cost into real-valued fourth order functions.






\section{Problem Statement and System Model \label{sec:CO-System Model}}

We consider a baseband MIMO system model with $N_{\text{T}}$ transmit
antennas and $N_{\text{R}}$ receive antennas. This MIMO model covers
SIMO and SISO channels. Let $s_{n}\left(k\right)$, $n=1,\ldots,N_{\text{T}}$
denote random independent data symbols for transmit antenna $n$ at
time $k$. $\{s_{n}\}$ typically belong to finite set ${\cal A}$.
Consistent with practical QAM systems, we assume that $s_{n}\left(k\right)$
are mutually independent and identically distributed (i.i.d.) with
variance $\sigma_{s}^{2}$ and fourth order kurtosis $\gamma_{s}<0$.

The data streams are transmitted over multipath MIMO channels denoted
by impulse responses $h_{j,n}\left(m\right)$, $n=1,\ldots,N_{\text{T}}$,
$j=1,\ldots,N_{\text{R}}$, $m=0,\ldots,L_{h}$, with delay spread of $L_{h}+1$ samples.
Assuming the MIMO channel is corrupted by i.i.d. additive white Gaussian
noise (AWGN) $v_{j}\left(k\right)$ with zero-mean and variance $\sigma_{v}^{2}$,
the output of $j$-th receiver can be expressed as 
\begin{eqnarray}
x_{j}\left(k\right) & = & \sum_{n=1}^{N_{\text{T}}}\sum_{m=0}^{L_{h}}h_{j,n}\left(m\right)s_{n}(k-m)+v_{j}\left(k\right). 
\end{eqnarray}

The goal of the MIMO receiver is to recover data stream $s_{n}\left(k\right)$,
$n=1,\ldots,N_{\text{T}}$ from the channel output $\{x_{j}\left(k\right)\}$.
Following most of the works in this area, we shall focus only on linear
FIR equalizers. Given $N_{\text{R}}$ sub-channel outputs, we apply
a $N_{\text{T}}\times N_{\text{R}}$ MIMO equalizer with parameter vector 
\[
\mathbf{w}_{i,j}=\left[w_{i,j}\left(0\right)\: w_{i,j}\left(1\right)\ldots w_{i,j}\left(L_{w}\right)\right]^{T}
\]
where $L_{w}$ is considered as the order of individual equalizer and $i=1,\ldots,N_{\text{T}}$, $j=1,\ldots,N_{\text{R}}$.
If the channel is flat fading (or non-frequency-selective), then we
have a degenerated problem of blind source separation for which $L_{h}=0$.
The FIR blind MIMO equalizer then degenerates into a linear source
separation matrix $\mathbf{W}=[w_{i,j}(0)]$.

The linear receiver output has $N_{\text{T}}$ parallel output streams
denoted by 
\begin{align}
y_{i}\left(k\right) &=\sum_{j=1}^{N_{\text{R}}}w_{i,j}^{*}\left(k\right)\circledast x_{j}(k) \nonumber \\
 &=\sum_{j=1}^{N_{\text{R}}}\sum_{\ell=0}^{L_{w}}w_{i,j}^{*}\left(\ell\right)x_{j}(k-\ell)
\end{align}

where $\circledast$ denote the convolution. Our objective is to optimize
the linear receiver parameters $\{w_{i,j}(n)\}$ such that the source
data symbols are recovered by $y_{i}\left(k\right)$ without interference
as follows 
\[
y_{i}\left(k\right)=e^{j\phi_{i}}s_{q_{i}}(k-k_{i})+\mbox{residual noise}.
\]
Note that $\phi_{i}$ is a phase ambiguity that is inherent to the
blind signal recovery which cannot be resolved without additional
information; and $k_{i}$ is the output delay of the $q_{i}-$th signal
recovery that does not affect the receiver performance. Upon optimum
convergence, a simple memoryless decision device $\mbox{dec}(\cdot)$
can be applied, and at high signal to noise ratio (SNR), we have 
\[
\hat{s}_{i}(k-k_{i})=\mbox{dec}\left[y_{i}(k)\right]=s_{q_{i}}(k-k_{i}).
\]

For convenience of notation, we further define 
\begin{align*}
\mathbf{w}_{i} & =\left[\mathbf{w}_{i,1}^{T}\:\mathbf{w}_{i,2}^{T}\ldots\mathbf{w}_{i,N_{\text{R}}}^{T}\right]^{T},\\
\mathbf{x}_{j}\left(k\right) & =\left[x_{j}\left(k\right)\: x_{j}\left(k-1\right)\ldots x_{j}\left(k-L_{w}\right)\right]^{T},\\
\mathbf{x}\left(k\right) & =\left[\mathbf{x}_{1}^{T}\left(k\right)\:\mathbf{x}_{2}^{T}\left(k\right)\ldots\mathbf{x}_{N_{\text{R}}}^{T}\left(k\right)\right]^{T}.
\end{align*}
With these notations, we can write 
\begin{align}
y_{i}\left(k\right) & =\sum_{j=1}^{N_{\text{R}}}\mathbf{w}_{i,j}^{H}\mathbf{x}_{j}(k)=\mathbf{w}_{i}^{H}\mathbf{x}(k)\nonumber \\
 & =\sum_{n=1}^{N_{\text{T}}}\underbrace{\left(\sum_{j=1}^{N_{\text{R}}}w_{i,j}^{*}\left(k\right)\circledast h_{j,n}\left(k\right)\right)}_{\mbox{combined response }c_{i,n}^{*}\left(k\right)}\circledast s_{n}(k) \nonumber\\
& +\underbrace{\sum_{j=1}^{N_{\text{R}}}w_{i,j}^{*}\left(k\right)\circledast v_{j}(k)}_{\mbox{Gaussian noise \ensuremath{\eta_{i}\left(k\right)}}}\\
 & =\sum_{n=1}^{N_{\text{T}}}c_{i,n}^{*}\left(k\right)\circledast s_{n}(k)+\eta_{i}\left(k\right).
\end{align}
In MIMO blind equalization, we would like to find equalizer parameter
vectors $\mathbf{w}_{i}$ such that the combined (channel-equalizer)
response is free of inter-symbol and inter-channel interferences 
\begin{equation}
c_{i,n}\left(k\right)=\begin{cases}
e^{j\phi_{i}} & \text{for }k=k_{i},\: n=q_{i}\\
0 & \text{otherwise}.
\end{cases}
\end{equation}

\section{Constant Modulus Algorithm for MIMO Equalization \label{sec:CO-CMA}}

In this section, we provide a real-valued representation of the conventional batch CMA cost. This representation is one of the keys to reduce parameter space compared to the work of \cite{Maricic2003} in which the formulation is based on complex values. This representation is also applied to other blind channel equalization and blind source separation costs and will be shown in the latter sections. 
\subsection{CMA for single source recovery}

For the sake of clarity, we discuss the CMA cost as an example of
selecting blind cost. To recover a particular source, the CMA cost
function for the $i$-th equalizer output sequence is defined as 
\begin{equation}
J_{\text{b},i}=J_{\text{cma},i}=E\left[\left(|y_{i}(k)|^{2}-R_{2}\right)^{2}\right]
\label{eq: Constant Modulus Algorithm}
\end{equation}
 where $R_{2}=\frac{E\left[\left|s_{i}\left(k\right)\right|^{4}\right]}{E\left[\left|s_{i}\left(k\right)\right|^{2}\right]}$.
The CMA cost can be represented as a function of the equalizer coefficients
and the channel output statistics \cite{Han2010}. 
Since the formulation here deals with single source, the source index $i$ is omitted.

Let $\text{Re}\left\{\boldsymbol{x}\right\}$ and $\text{Im}\left\{\boldsymbol{x}\right\}$
denote the real and imaginary parts of $\boldsymbol{x}$. Now define
$\mathbf{u}=\left[\begin{array}{c}
\text{Re}\left\{\mathbf{w}\right\}\nonumber\\
\text{Im}\left\{\mathbf{w}\right\}
\end{array}\right]$, $\mathbf{x}_{\text{r}}(k)=\left[\begin{array}{c}
\text{Re}\left\{ \mathbf{x}(k)\right\} \\
\text{Im}\left\{ \mathbf{x}(k)\right\} 
\end{array}\right]$ and $\mathbf{x}_{\text{i}}(k)=\left[\begin{array}{c}
\text{Im}\left\{ \mathbf{x}(k)\right\} \\
-\text{Re}\left\{ \mathbf{x}(k)\right\} 
\end{array}\right]$.
We have the following relationship 
\begin{eqnarray}
\text{Re}\left\{ y\left(k\right)\right\}  & = & \mathbf{u}^{T}\mathbf{x}_{\text{r}}(k),\\
\text{Im}\left\{ y\left(k\right)\right\}  & = & \mathbf{u}^{T}\mathbf{x}_{\text{i}}(k).
\end{eqnarray}
As a result, we have 
\begin{align}
\begin{split}
\left|y\left(k\right)\right|^{2} & =  \text{Re}^{2}\left\{ y\left(k\right)\right\} +\text{Im}^{2}\left\{ y\left(k\right)\right\} \\
 & = \mathbf{u}^{T}\mathbf{x}_{\text{r}}(k)\mathbf{x}_{\text{r}}(k)^{T}\mathbf{u}+\mathbf{u}^{T}\mathbf{x}_{\text{i}}(k)\mathbf{x}_{\text{i}}(k)^{T}\mathbf{u}\\
 & = \mathbf{u}^{T}\underbrace{\left[\mathbf{x}_{\text{r}}(k)\mathbf{x}_{\text{r}}^{T}(k)+\mathbf{x}_{\text{i}}(k)\mathbf{x}_{\text{i}}^{T}(k)\right]}_{\bar{\mathbf{X}}_{k}}\mathbf{u}.
\end{split}
\end{align}
By denoting the rank-1 matrix $\bar{\mathbf{X}}_{k}=\mathbf{x}_{\text{r}}(k)\mathbf{x}_{\text{r}}^{T}(k)+\mathbf{x}_{\text{i}}(k)\mathbf{x}_{\text{i}}^{T}(k)$
and $\mathbf{U}=\mathbf{u}\mathbf{u}^{T}$, we have 
\begin{equation}
\left|y\left(k\right)\right|^{2}=\text{Tr}\left(\bar{\mathbf{X}}_{k}\mathbf{U}\right).
\end{equation}
Note that both $\bar{\mathbf{X}}_{k}$ and $\mathbf{U}$ are symmetric
of dimension $2N\times2N$ with $N=L_{w}+1$. In other words, $\bar{\mathbf{X}}_{k}$
and $\mathbf{U}$ can be mapped by the $\text{svec\ensuremath{\left(\cdot\right)}}$
operator to lower dimensional subspace $\mathbb{R}^{N\left(2N+1\right)}$
after ignoring redundant entries. Furthermore, we sort the entries
in the usual lexicographic order. We define the following operators
and vectors: 
\begin{itemize}
\item For a symmetric $2N\times2N$ matrix $\mathbf{X}$, we define $\text{svec\ensuremath{\left(\cdot\right)}}$
operator and its reverse operator $\mbox{svec}^{-1}\left(\cdot\right)$
as 
\begin{align}
\begin{split}
\text{svec}\left(\mathbf{X}\right) & =\left[ X_{1,1}\:2X_{1,2}\ldots\:2X_{1,2n} \:X_{2,2}\right. \\ 
& \quad \left. 2X_{2,3}\ldots2X_{2,2N}\ldots X_{2N,2N}\right]^{T} \in\mathbb{R}^{N(2N+1)} , \\
\mathbf{X} & =\mbox{svec}^{-1}\left(\text{svec}\left(\mathbf{X}\right)\right).
\end{split}
\end{align}

\item For a vector $\mathbf{u}\in\mathbb{R}^{2N}$, we define $\mathbf{v}=\mbox{qvec}(\mathbf{u})$
as vector whose entries are all second order (quadratic) terms $\left\{ u_{i}u_{j}\right\} $
and sorted in the usual lexicographic order. 
\begin{align}
\begin{split}
\mathbf{v}&=\mbox{qvec}(\mathbf{u}) \\
&=[u_{1}u_{1}\: u_{1}u_{2}\ldots u_{1}u_{2n}\: u_{2}u_{2}\: u_{2}u_{3}\:\ldots \\
&\quad\;\; u_{2}u_{2N}\ldots u_{2N}u_{2N}]^{T}.
\end{split}
\end{align}
There is one to one mapping between elements of $\mathbf{v}$ and
$\mathbf{U}=\mathbf{u}\mathbf{u}^{T}$ as $\mathbf{v}$ consists of
all upper triangle elements of $\mathbf{U}$. Nevertheless, the reverse
operator needs further consideration. Given an arbitrary $\mathbf{v}\in\mathbb{R}^{2N^{2}+N+1}$,
we can form the corresponding $\mathbf{U}$, which is not necessarily
a rank 1 matrix. Therefore, we define the reverse operator with approximation
$\mbox{qvec}^{-1}(\mathbf{v})$ as follows: First, using $\mathbf{v}$
to form the corresponding $\mathbf{U}$; Next, find the rank 1 approximation
of $\mathbf{U}$ by its maximum eigenvalue $\lambda_{\mathbf{U}}$
and the corresponding eigenvector $\mathbf{u}_{\text{eig}}$. The
resulting matrix is ${\lambda_{\mathbf{U}}}\mathbf{u}_{\text{eig}}\mathbf{u}^T_{\text{eig}}$. 
\end{itemize}
With the above definitions, the output power can be rewritten as 
\begin{align}
\begin{split}
E\left[\left|y\left(k\right)\right|^{2}\right] & = E\left[\mathbf{v}^{T}\text{svec}\left(\bar{\mathbf{X}}_{k}\right)\right]= \mathbf{v}^{T}\text{svec}\left(E\left[\bar{\mathbf{X}}_{k}\right]\right)\\
 & = \mathbf{v}^{T}\underset{\mathbf{b}}{\underbrace{E\left[\text{svec}\left(\bar{\mathbf{X}}_{k}\right)\right]}} = \mathbf{v}^{T}\mathbf{b}.
\end{split}
\end{align}
Similarly, the fourth order moment of equalizer output is 
\begin{eqnarray}
E\left[\left|y\left(k\right)\right|^{4}\right] & = & \mathbf{v}^{T}\underset{\mathbf{C}}{\underbrace{E\left[\text{svec}\left(\bar{\mathbf{X}}_{k}\right)\text{svec}^{T}\left(\bar{\mathbf{X}}_{k}\right)\right]}}\mathbf{v}\nonumber \\
 & = & \mathbf{v}^{T}\mathbf{C}\mathbf{v}.
\end{eqnarray}
Therefore, the CMA cost can be written into a fourth order function
of $\mathbf{u}$ or a quadratic function of $\mathbf{v}$ as 
\begin{eqnarray}
J_{\text{cma}} & = & \mathbf{v}^{T}\mathbf{C}\mathbf{v}-2R_{2}\mathbf{b}^{T}\mathbf{v}+R_{2}^{2}.\label{eq:CMA cost matrix form-v}
\end{eqnarray}
The equalizer can be found by minimizing the polynomial function $J_{\text{cma}}$.

\subsection{Multi-source recovery}

To recover multiple input source signals of a MIMO channel, we must
form several single stream receivers $\{\mathbf{w}_{i}\}$ to generate
multiple receiver output streams $\{y_{i}(k)\}$ that will represent
distinct source signals $\{s_{q_{i}}(k)\}$. Typically, a blind algorithm
is adept at recovering one of the possible source signals. However,
because of the inherent properties of blind equalization, the receiver
is unable to ascertain which signal may be recovered {\em a priori}.
As a result, multiple single stream receivers may recover duplicate
signals and miss some critical source signals. To avoid duplicate
convergence, proper initialization of each $\mathbf{w}_{i}$ may help
lead to a diversified convergence as needed. However, there is unfortunately
no guarantee that different initial values of $\mathbf{w}_i$ will lead to $q_{i}\not=q_{j}$.
In order to ensure that different receiver vectors $\mathbf{w}_{i}$
extract distinct signals, common criteria rely on the prior information
that the source signals are i.i.d. and mutually independent.

Let $J_{\text{b},i}$ denote a blind cost function, for examples,
CMA, SWA, or MED, to recover the source $i$, the cost function for multiple
blind source recovery (BSR) is \cite{Papadias1997,Li1998} 
\begin{equation}
J_{\text{bsr}}=\sum_{i=1}^{N_{\text{T}}}J_{\text{b},i}+\lambda_{\text{cr}}\cdot\sum_{i\neq j}^{N_{\text{T}}}J_{\text{cr},i,j}\label{eq:Blind Cost with source separation-simul}
\end{equation}
where $\lambda_{\text{cr}}$ is a positive scalar and $J_{\text{cr},i,j}$
accounts for the correlation of the receiver outputs $i$ and $j$.
The minimization of $J_{\text{cr},i,j}$ ensures that the equalizers
converge to different solutions corresponding to different sources.
One can recover all sources either simultaneously or sequentially
to reduce the parameter size.

In sequential source recovery, we start by extracting the first source
by minimizing one single CMA cost. Assume the sources up to $j-1$
are equalized and separated, we can minimize the following cost for
separating the source $j$. 
\begin{equation}
J_{\text{bsr},j}=J_{\text{b},j}+\lambda_{\text{cr}}\sum_{i=1}^{j-1}J_{\text{cr},i,j}.\label{eq:Blind Cost with source separation-sequential}
\end{equation}
In this work, we consider the sequential approach. $J_{\text{cr},i,j}$
can be chosen as the sum of cross-cumulants \cite{Li1998} or as the
sum of cross-correlations \cite{Papadias1997}. Since the use of cross-cumulant
is for coarse step separation \cite{Li1998} and may lead to poor
convergence, here, we use the cross-correlation 
\begin{equation}
J_{\text{cr},i,j}=\sum_{l=-\delta}^{\delta}\left|E\left[y_{i}\left(k\right)y_{j}^{*}\left(k-l\right)\right]\right|^{2}
\label{eq:Source separation with cumulant}
\end{equation}
where $\delta$ is an integer parameter to be chosen during system
design. $J_{\text{cr},i,j}$ can also be written as a fourth order
function of $\mathbf{u}_{j}$ or a second order function of $\mathbf{v}_{j}=\mbox{qvec}(\mathbf{u}_{j})$.
Noticing that $\mathbf{u}_{i}$, for $i<j$, are already known, we
can write 
\begin{equation}
\sum_{i=1}^{j-1}J_{\text{cr},i,j}=\mathbf{q}_{j}^{T}\mathbf{v}_{j}
\end{equation}
where $\mathbf{q}_{j}$ is a vector formed by the previously calculated
equalizers and received statistics. The detailed calculation of $\mathbf{q}_{j}$
is given in the Appendix A.

Eventually, the cost in \eqref{eq:Blind Cost with source separation-sequential}
when using CMA and cross-correlation can be written as a fourth order
(cost) function of equalizer parameter with zero odd-order coefficients
similar to $J_{\text{cma}}$: 
\begin{align}
J_{\text{bsr},j} & =f(\mathbf{u}_{j})\nonumber \\
 & =\mathbf{v}_{j}^{T}\mathbf{C}\mathbf{v}_{j}-2R_{2}\mathbf{b}^{T}\mathbf{v}_{j}+R_{2}^{2}+\lambda_{\text{cr}}\mathbf{q}_{j}^{T}\mathbf{v}_{j}\nonumber \\
 & =\mathbf{v}_{j}^{T}\mathbf{C}\mathbf{v}_{j}-\left(2R_{2}\mathbf{b}^{T}-\lambda_{\text{cr}}\mathbf{q}_{j}^{T}\right)\mathbf{v}_{j}+R_{2}^{2}.
\end{align}
In the following section, we discuss the method to find global minima of the functions of this type.

\section{Semidefinite programming approach to minimize a fourth order polynomial
\label{sec:CO-SDP of 4th order function}}

\subsection{General formulation}

General formula to find the global minimum of a non-negative polynomial
is discussed in \cite{Shor1997,Parrilo2000}. Here, we restrict the
formulation for the fourth order polynomial $f\left(\mathbf{u}\right)$
with $\mathbf{u}\in\mathbb{R}^{2N}$. The minimization of $f\left(\mathbf{u}\right)$
is equivalent to 
\begin{align}
\text{max} & \qquad\tau\label{eq:max tau for all u}\\
\text{s.t.} & \qquad f\left(\mathbf{u}\right)-\tau\ge0\text{ for all }\mathbf{u}.\nonumber 
\end{align}

This optimization is equivalent to lifting the horizontal hyperplane
created by $\tau$ until it lies immediately beneath the hypersurface
$f\left(\mathbf{u}\right)$. The intersection points that the hyperplane
and the hypersurface are the global minima of $f\left(\mathbf{u}\right)$.
This problem is convex with a convex cost and linear constraints in
$\tau$. The problem, however, still a hard problem to solve since
the number of constraints is infinite. Although it is hard to find
the optimal solution to this problem, we can modify and simplify the
problem by narrowing down the search space. Since $f\left(\mathbf{u}\right)-\tau$
is non-negative, we are looking for a representation of $f\left(\mathbf{u}\right)-\tau$ as a sum of square. 
Following the work of \cite{Maricic2003}, we define two convex cones of fourth order polynomials of $\mathbf{u}$,
$\mathcal{C}$ and $\mathcal{D}$: 
\begin{align*}
\mathcal{C} & =\left\{ \left.g\right|g(\mathbf{u})\text{ is real-valued fourth order polynomial of } \mathbf{w} \right. \\ 
			& \qquad \qquad \left. \text{ and }g(\mathbf{u})\ge0,\forall\mathbf{u}\right\} ,\\
\mathcal{D} & =\left\{ \left. g\right|g(\mathbf{u})=\sum_{i}g_{i}^{2}(\mathbf{u})\text{ with each } \right. \\ 
      & \qquad \left. g_{i}(\mathbf{u})\text{ is real-valued second-order polynomial of }\mathbf{u}\right\} .
\end{align*}

We can rewrite the optimization problem in \eqref{eq:max tau for all u}
in term of $\mathcal{C}$ as 
\begin{align}
\text{max } & \qquad\tau\label{eq:max tau for f-tau in C}\\
\text{s.t.} & \qquad f\left(\mathbf{u}\right)-\tau\in\mathcal{C}.\nonumber 
\end{align}
Since the problem in \eqref{eq:max tau for f-tau in C} is hard to
solve, we can follow \cite{Maricic2003} and narrow down our feasible
solution from $\mathcal{C}$ to a more restrictive set $\mathcal{D}$
\begin{align}
\text{max } & \qquad\tau\label{eq:max tau for f-tau in D}\\
\text{s.t. } & \qquad f\left(\mathbf{u}\right)-\tau\in\mathcal{D}.\nonumber 
\end{align}
This problem can be cast into a convex semi-definite programming to
be solved efficiently \cite{Nesterov2000}.

We recall the following lemma, which is a simplified version of the
results in \cite{Nesterov2000}. 
\begin{lem}
\label{lem:forth-order polynomial with PSD}Given any forth-order
polynomial $g\left(\mathbf{u}\right)$, the following relation holds:
\begin{align}
 & g\left(\mathbf{u}\right)\in\mathcal{D}\Longleftrightarrow g\left(\mathbf{u}\right)=\bar{\mathbf{u}}^{T}\mathbf{G}\bar{\mathbf{u}}\nonumber \\
 & \text{for some symmetric matrix }\mathbf{G}\succcurlyeq\mathbf{0}
\end{align}
where $\bar{\mathbf{u}}=\left[\mathbf{u}_{(2)}\:\mathbf{u}_{(1)}\: u_{(0)}\right]^{T}$
and $u_{(0)}=1$, $\mathbf{u}_{(1)}=\mathbf{u}$, $\mathbf{u}_{(2)}=\mbox{\text{qvec}}\left(\mathbf{u}\right)$. 
\end{lem}
Therefore, we obtain an equivalent formulation of \eqref{eq:max tau for f-tau in D}
as 
\begin{align}
\text{max } & \qquad\tau\nonumber \\
\text{s.t. } & \qquad f\left(\mathbf{u}\right)-\tau=\bar{\mathbf{u}}^{T}\mathbf{G}\bar{\mathbf{u}}\text{ for all }\bar{\mathbf{u}}\nonumber \\
\text{and } & \qquad\mathbf{G}\succcurlyeq0.\label{eq:SOS optimization}
\end{align}
The problem can be solved efficiently using semi-definite programming
\cite{Vandenberghe1996}. Solving \eqref{eq:SOS optimization} for
$\tau$, $\mathbf{G}$, we obtain the global solution for this problem.

In general, the solutions for \eqref{eq:max tau for f-tau in C} and
\eqref{eq:max tau for f-tau in D} are not the same. Define $\tau_{\mathcal{C}}$
and $\tau_{\mathcal{D}}$ the solutions of \eqref{eq:max tau for f-tau in C},
\eqref{eq:max tau for f-tau in D}, respectively. Since, \eqref{eq:max tau for f-tau in D}
is a restrictive version of \eqref{eq:max tau for f-tau in C}, we
have $\tau_{\mathcal{C}}\le\tau_{\mathcal{D}}$. Nevertheless, simulation
tests in \cite{Parrilo2003} showed that the solutions of \eqref{eq:max tau for f-tau in C}
and \eqref{eq:max tau for f-tau in D} are nearly identical for arbitrary
polynomial cost functions $f(\mathbf{u})$. In the next subsection,
we will show that they are actually identical in the case of CMA cost.

\subsection{Minimization of fourth order cost functions without odd-order coefficients}

Now, we would like to specialize the algorithm for solving CMA cost.
Therefore, we focus on the functions that are fourth order without
odd-order entries.

Define 
\begin{align*}
\mathcal{E}&=\left\{ \left.g\right|g(\mathbf{u})\text{ is real-valued polynomial of }\mathbf{u} \right. \\ 
& \qquad \left.  \text{ with only even-order coefficients}\right\} .
\end{align*}
The CMA cost and other mentioned costs in this work belong to the
function set $\mathcal{C}\cap\mathcal{E}$. The minimization problem
is simply 
\begin{align}
\text{max } & \qquad\tau\label{eq:max tau for f-tau in C and E}\\
\text{s.t.} & \qquad f\left(\mathbf{u}\right)-\tau\in\mathcal{C}\cap\mathcal{E}.\nonumber 
\end{align}
As the problem in \eqref{eq:max tau for f-tau in C and E} remains
hard to solve, we reduce our feasible solution from $\mathcal{\mathcal{C}\cap\mathcal{E}}$
to $\mathcal{D\cap\mathcal{E}}$ as 
\begin{align}
\text{max } & \qquad\tau\label{eq:max tau for f-tau in D and E}\\
\text{s.t. } & \qquad f\left(\mathbf{u}\right)-\tau\in\mathcal{D\cap\mathcal{E}}.\nonumber 
\end{align}

We can arrive at the following proposition: 
\begin{prop}
The solutions of the problems in \eqref{eq:max tau for f-tau in C and E}
and \eqref{eq:max tau for f-tau in D and E} are identical. 
\end{prop}
\begin{proof} Let $\bar{\tau}$ be the optimal solution for \eqref{eq:max tau for f-tau in C and E}.
Since $f\left(\mathbf{u}\right)-\bar{\tau}$ is a quadratic function
of $\mathbf{v}$ and is non-negative for all $\mathbf{v}$, it can
always be written as a sum of squares of polynomials \cite{Rudin2000}.
In other words, the solution of \eqref{eq:max tau for f-tau in C and E}
and \eqref{eq:max tau for f-tau in D and E} are identical. \end{proof}
Thus, unlike the general formulation, the restrictive problem in \eqref{eq:max tau for f-tau in D and E}
does not introduce any gap. We can now refine Lemma \ref{lem:forth-order polynomial with PSD}
into the following proposition: 
\begin{prop}
\label{lem:forth-order polynomial with PSD simplify}Given any forth-order
polynomial $g\left(\mathbf{u}\right)$ without odd-order coefficients,
the following relationship holds: 
\begin{align}
 & g\left(\mathbf{u}\right)\in\mathcal{D}\cap\mathcal{E}\Longleftrightarrow g\left(\mathbf{u}\right)=\tilde{\mathbf{u}}^{T}\tilde{\mathbf{G}}\tilde{\mathbf{u}}\nonumber \\
 & \text{for some symmetric matrix }\tilde{\mathbf{G}}\succcurlyeq\mathbf{0}\label{eq:u tilde definition}
\end{align}
where $\tilde{\mathbf{u}}=\left[\mathbf{u}_{(2)}\: u_{(0)}\right]^{T}$,
$u_{(0)}=1$, and $\mathbf{u}_{(2)}=\mbox{qvec}\left(\mathbf{u}\right)$. 
\end{prop}
\begin{proof} From Lemma \ref{lem:forth-order polynomial with PSD},
we find $\mathbf{G}$ such that $g\left(\mathbf{u}\right)=\bar{\mathbf{u}}^{T}\mathbf{G}\bar{\mathbf{u}}$
and $\mathbf{G}\succcurlyeq0$. We partition $\mathbf{G}$ as 
\begin{equation}
\mathbf{G}=\left[\begin{array}{ccc}
\mathbf{G}_{22} & \mathbf{G}_{21} & \mathbf{G}_{20}\\
\mathbf{G}_{21}^{T} & \mathbf{G}_{11} & \mathbf{G}_{10}\\
\mathbf{G}_{20}^{T} & \mathbf{G}_{10}^{T} & \mathbf{G}_{00}
\end{array}\right]
\end{equation}
where $\mathbf{G}_{ij}$ corresponds to the coefficients for the product
terms between entries of $\mathbf{u}_{(i)}$ and $\mathbf{u}_{(j)}$.
Consequently, the function $g\left(\mathbf{u}\right)$ can be rewritten
as 
\begin{align}
g\left(\mathbf{u}\right) & =\mathbf{u}_{(2)}^{T}\mathbf{G}_{22}\mathbf{u}_{(2)}+\mathbf{u}_{(2)}^{T}\mathbf{G}_{21}\mathbf{u}_{(1)}+\mathbf{u}_{(2)}^{T}\mathbf{G}_{20}\nonumber \\
 & +\mathbf{u}_{(1)}^{T}\mathbf{G}_{21}^{T}\mathbf{u}_{(2)}+\mathbf{u}_{(1)}^{T}\mathbf{G}_{11}\mathbf{u}_{(1)}+\mathbf{u}_{(1)}^{T}\mathbf{G}_{10}\nonumber \\
 & +\mathbf{G}_{20}^{T}\mathbf{u}_{(2)}+\mathbf{G}_{10}^{T}\mathbf{u}_{(1)}+\mathbf{G}_{00}\nonumber \\
 & =\mathbf{u}_{(2)}^{T}\mathbf{G}_{22}\mathbf{u}_{(2)}+2\mathbf{u}_{(2)}^{T}\mathbf{G}_{20}+\mathbf{u}_{(1)}^{T}\mathbf{G}_{11}\mathbf{u}_{(1)}+\mathbf{G}_{00}\nonumber \\
 & +2\mathbf{u}_{(2)}^{T}\mathbf{G}_{21}\mathbf{u}_{(1)}^{T}+2\mathbf{u}_{(1)}^{T}\mathbf{G}_{10}.\label{eq:uGu expansion}
\end{align}
Since $g\left(\mathbf{u}\right)$ does not have odd-order entries,
we have 
\begin{align}
\mathbf{u}_{(2)}^{T}\mathbf{G}_{21}\mathbf{u}_{(1)}^{T} & =0,\\
\mathbf{u}_{(1)}^{T}\mathbf{G}_{10} & =0.
\end{align}
We can rewrite $g\left(\mathbf{u}\right)=\bar{\mathbf{u}}^{T}\bar{\mathbf{G}}\bar{\mathbf{u}}$,
where 
\begin{align*} \bar{\mathbf{G}}=\left[\begin{array}{ccc}
\mathbf{G}_{22} & 0 & \tilde{\mathbf{G}}_{20}\\
0 & 0 & 0\\
\mathbf{G}_{20}^{T} & 0 & \mathbf{G}_{00}
\end{array}\right]
\end{align*}
 and 
\begin {align*}
\tilde{\mathbf{G}}_{20}=\mathbf{G}_{20}+\text{svec}\left(\mathbf{G}_{11}\right)/2.
\end{align*}
This is equivalent to 
\begin {align*}
g\left(\mathbf{u}\right)=\tilde{\mathbf{u}}^{T}\tilde{\mathbf{G}}\tilde{\mathbf{u}}
\end{align*}
where $\tilde{\mathbf{G}}=\left[\begin{array}{cc}
\mathbf{G}_{22} & \tilde{\mathbf{G}}_{20}\\
\tilde{\mathbf{G}}_{20}^{T} & \mathbf{G}_{00}
\end{array}\right]\succcurlyeq0$ since $g\left(\mathbf{u}\right)\ge0$ by definition. \end{proof}

\subsection{Semidefinite programing solution}

\noindent Following Proposition \ref{lem:forth-order polynomial with PSD simplify},
the optimization problem in \eqref{eq:max tau for f-tau in D and E}
is equivalent to solving 
\begin{align}
\text{max } & \qquad\tau\nonumber \\
\text{s.t. } & \qquad f\left(\mathbf{u}\right)-\tau=\tilde{\mathbf{u}}^{T}\tilde{\mathbf{G}}\tilde{\mathbf{u}}\nonumber \\
\text{and } & \qquad\tilde{\mathbf{G}}\succcurlyeq0\label{eq:SOS optimization-2}
\end{align}
where $f\left(\mathbf{u}\right)\in\mathcal{D}\cap\mathcal{E}$, i.e.,
$f\left(\mathbf{u}\right)=\mathbf{u}_{(2)}^{T}\mathbf{A}_{22}\mathbf{u}_{(2)}+2\mathbf{A}_{20}^{T}\mathbf{u}_{(2)}+\mathbf{A}_{00}$,
and $\tilde{\mathbf{u}}$ is defined as in \eqref{eq:u tilde definition}.
Here, $\mathbf{A}_{ij}$ corresponds to the coefficients for the product
terms between entries of $\mathbf{u}_{(i)}$ and $\mathbf{u}_{(j)}$
and $u_{(0)}=1$. The optimization in (\ref{eq:SOS optimization-2})
can be recast into 
\begin{align}
\begin{split}
\text{max } & \qquad\tau\\
\text{s.t. } & \qquad\mathbf{u}_{(2)}^{T}\mathbf{A}_{22}\mathbf{u}_{(2)}=\mathbf{u}_{(2)}^{T}\mathbf{G}_{22}\mathbf{u}_{(2)}, \\
 & \qquad\mathbf{A}_{20}^{T}\mathbf{u}_{(2)}=\mathbf{G}_{20}\mathbf{u}_{(2)}, \\
 & \qquad\mathbf{A}_{00}-\tau=\mathbf{G}_{00} \\
\text{and } & \qquad\tilde{\mathbf{G}}\succcurlyeq0.
\label{eq:SOS optimization-3}
\end{split}
\end{align}

Similar to \cite{Maricic2003}, we let $1\le i\le j\le l\le m\le2N$
and define $\mathcal{P}_{i,j,l,m}$ as the set of all distinct 4-tuples
that are permutation of $\left(i,j,l,m\right)$. Now define a subset
of $\mathcal{P}_{i,j,l,m}$ as 
\[
\mathcal{Q}_{i,j,l,m}=\left\{ \left(i,j,l,m\right),\left(i,j,l,m\right)\in\mathcal{P}_{i,j,l,m}\mbox{ and }i\le j,\: l\le m\right\} .
\]
$f\left(\mathbf{u}\right)-\tau$ can be written into two ways as follows
\begin{align}
\begin{split}
& f\left(\mathbf{u}\right)-\tau  =\tilde{\mathbf{u}}^{T}\tilde{\mathbf{G}}\tilde{\mathbf{u}} \\
 & =\sum_{1\le i\le j\le l\le m\le2N}\left(\sum_{\left(i',j',l',m'\right)\in\mathcal{Q}_{i',j',l',m'}}\mathbf{G}_{(i',j')(l',m')}\right)\\
&u_{i}u_{j}u_{l}u_{m} +\sum_{1\le i\le j\le2N}\left(2\mathbf{G}_{(i,j)}\right)u_{i}u_{j}+\mathbf{G}_{00}
\label{eq:fu-tau 1}
\end{split}
\end{align}
and 
\begin{align}
\begin{split}
&f\left(\mathbf{u}\right)-\tau  =\mathbf{u}_{(2)}^{T}\mathbf{A}_{22}\mathbf{u}_{(2)}+2\mathbf{A}_{20}^{T}\mathbf{u}_{(2)}+\mathbf{A}_{00}-\tau \\
 & =\sum_{1\le i\le j\le l\le m\le2N}\left(\sum_{\left(i',j',l',m'\right)\in\mathcal{Q}_{i',j',l',m'}}\mathbf{A}_{(i',j')(l',m')}\right)\\
&u_{i}u_{j}u_{l}u_{m} +\sum_{1\le i\le j\le2N}\left(2\mathbf{A}_{(i,j)}\right)u_{i}u_{j}+\mathbf{A}_{00}-\tau\label{eq:fu-tau 2}
\end{split}
\end{align}
where $\mathbf{G}_{(i,j)(l,m)}$ and $\mathbf{A}_{(i,j)(l,m)}$ represents
the entries of $\mathbf{G}_{22}$ and $\mathbf{A}_{22}$, respectively,
for the product term between entries $u_{i}u_{j}$ and $u_{l}u_{m}$
in $\mathbf{u}_{(2)}$. Similarly, $\mathbf{G}_{(i,j)}$ and $\mathbf{A}_{(i,j)}$
denote the entries of $\mathbf{G}_{20}$ and $\mathbf{A}_{20}$, respectively,
for product term $u_{i}u_{j}$ in $\mathbf{u}_{(2)}$.

The optimization in (\ref{eq:SOS optimization-3}) can be rewritten as
\begin{align}
\qquad\text{max } & \tau\nonumber \\
\text{s.t. } & \qquad\left(\sum_{\left(i',j',l',m'\right)\in\mathcal{Q}_{i,j,l,m}}\mathbf{G}_{(i',j')(l',m')}\right)\\&=\left(\sum_{\left(i',j',l',m'\right)\in\mathcal{Q}_{i,j,l,m}}\mathbf{A}_{(i',j')(l',m')}\right),\nonumber \\
 & \qquad\qquad\qquad\forall \: 1\le i\le j\le l\le m\le2N,\nonumber \\
 & \mathbf{G}_{(ij)}=\mathbf{A}_{(ij)},\quad\forall \: 1\le i\le j\le2N,\nonumber \\
 & \mathbf{A}_{00}-\tau=\mathbf{G}_{00},\nonumber \\
 & \tilde{\mathbf{G}}\succcurlyeq0.\label{eq:SOS optimization-4}
\end{align}
This problem is convex with linear and semi-definite constraints.
Therefore, it can be solved efficiently using available optimization
tools such as Sedumi or SDPT-3 \cite{Tutuncu2003}.

The proposed convex optimization is a real-valued and a simplified version
of the complex-valued problem in \cite{Maricic2003}. By modifying
the approach given in \cite{Maricic2003}, our method has two numerical
benefits that lead to lower complexity. First, it exploits the special
characteristic of fourth order statistics blind costs that have zero
odd-order coefficients. As a result, the unknown parameter space is
reduced. Second, the real-valued formulation captures all necessary
information for equalization. Because in our real-valued formulation,
the matrix $\mathbf{G}$ does not increase in size, the real-valued
formulation therefore further reduces the number of parameters to half.


\subsection{Post-processing}

After reaching the global solution $\mathbf{G}_{\text{opt}}$ of \eqref{eq:SOS optimization}
via SDP, the result must be translated and mapped back to the original
blind receiver parameter space. Here, we describe such post-processing
procedure.

If the global optimum is achieved, there exists $\tilde{\mathbf{u}}$
such that 
\begin{equation}
\begin{cases}
\tilde{\mathbf{u}}^{T}\mathbf{G}_{\text{opt}}\tilde{\mathbf{u}} & =0\\
\tilde{\mathbf{u}}^{T} & =[\tilde{\mathbf{u}}_{(2)}^{T}\;1]\\
\mbox{rank}\left(\tilde{\mathbf{U}}\right) & =1
\end{cases}
\end{equation}
where $\tilde{\mathbf{U}}$ is the symmetric matrix formed by the
elements of $\tilde{\mathbf{u}}_{(2)}$ inside $\mbox{qvec}^{-1}(\cdot)$.
Since $\mathbf{G}_{\text{opt}}$ is positive semidefinite, the solution
$\tilde{\mathbf{u}}$ must lie in the null space of $\mathbf{G}_{\text{opt}}$.
If the null space of $\mathbf{G}_{\text{opt}}$ has dimension 1, then
we already have $\tilde{\mathbf{u}}_{\text{opt}}$ uniquely. Otherwise,
we must look for $\tilde{\mathbf{u}}$ inside the eigenspace corresponding
to the smallest eigenvalues of $\mathbf{G}_{\text{opt}}$. Furthermore,
the last element of $\tilde{\mathbf{u}}_{\text{opt}}$ must equal
to 1. We now describe an iterative algorithm for real-valued data that was modified from the one proposed
in \cite{Maricic2003}: 
\begin{itemize}
\item Step 1. Initialization: pick a random $\tilde{\mathbf{u}}^{\left(0\right)}$.
Find the null space of $\mathbf{G}_{\text{opt}}$ by finding $\mathbf{V}$
that consists of eigenvectors of $\mathbf{G}_{\text{opt}}$ whose
eigenvalues are below a set threshold $\gamma$. 
\item Step 2. For a $k>0$, find the linear projection of $\tilde{\mathbf{u}}^{\left(k\right)}$
onto $\mathbf{V}$: 
\begin{equation}
\hat{\mathbf{u}}=\mathbf{V}\mathbf{V}^{T}\tilde{\mathbf{u}}^{(k)}.\label{eq:Gopt null space projection}
\end{equation}

\item Step 3. Normalize the last element to 1. Rescale $\hat{\mathbf{u}}$
as 
\begin{equation}
\check{\mathbf{u}}=\hat{\mathbf{u}}/\hat{u}_{0}
\end{equation}
where $\hat{u}_{0}$ is the last element of $\hat{\mathbf{u}}$. The
resulting vector $\check{\mathbf{u}}$ now is in null space of $\mathbf{G}_{\text{opt}}$
and its last element is 1. Let $\check{\mathbf{u}}_{2}$ be the vector
consisting of the first $2N^{2}+N$ elements of $\check{\mathbf{u}}$
as $\check{\mathbf{u}}=\left[\check{\mathbf{u}}_{2}\:1\right]^{T}$. 
\item Step 4. Calculate $\mathbf{u}=\mbox{qvec}^{-1}\left(\check{\mathbf{u}}_{2}\right)$,
note that this is an approximation operation related to rank 1 approximation.
Form $\tilde{\mathbf{u}}^{\left(k+1\right)}=\left[\left(\text{qvec}\left(\mathbf{u}\right)\right)^{T}\;1\right]^{T}$. 
\item Step 5. Repeat step 2 until converge. The equalizer is contained in
$\mathbf{u}$ as 
\[
\mathbf{u}=\left[\mbox{Re}\left(w^{T}\right)\:\mbox{Im}\left(w^{T}\right)\right]^{T}.
\]

\end{itemize}

Note that in step 3 of this existing post-processing, the last element
of $\check{\mathbf{u}}$ must be nonzero. In practice, this condition
may not always be true. To overcome this weakness, we can require
instead that the equalizer output have the same power as the transmitted
symbols. We introduce an additional gain $g$ on $\mathbf{u}$, (or
$g^{2}$ on $\mathbf{v}=\mbox{qvec}(\mathbf{u})$ ) such that $g\mathbf{u}$
produces the output with the same power as the transmitted symbol,

\begin{equation}
E\left[\left|y\left(k\right)\right|^{2}\right]=g^{2}\mathbf{v}^{T}\mathbf{b}=\sigma_{s}^{2}.
\end{equation}
Hence, 
\begin{equation}
g=\sqrt{\frac{\sigma_{s}^{2}}{\mathbf{v}^{T}\mathbf{b}}}.\label{eq:Rescale the equalizer}
\end{equation}
As a result, we have a more robust and new post-processing algorithm
as follows: 
\begin{itemize}
\item Step 1. Initialization: pick a random $\tilde{\mathbf{u}}^{\left(0\right)}$.
Find the null space of $\mathbf{G}_{\text{opt}}$ by finding matrix $\mathbf{V}$
which consists of eigenvectors of $\mathbf{G}_{\text{opt}}$
corresponding to the eigenvalues less than threshold $\gamma$. 
\item Step 2. For a $k>0$, find $\hat{\mathbf{u}}$ which is the linear
projection of $\tilde{\mathbf{u}}^{\left(k\right)}$ onto $\mathbf{V}$
as \eqref{eq:Gopt null space projection}. Let $\hat{\mathbf{u}}_{2}$
be the vector consisting of the first $2N^{2}+N$ elements of $\hat{\mathbf{u}}$
as $\hat{\mathbf{u}}=\left[\hat{\mathbf{u}}_{2}\:1\right]^{T}$. 
\item Step 3. Calculate $\mathbf{u}=\mbox{qvec}^{-1}\left(\hat{\mathbf{u}}_{2}\right)$. 
\item Step 4. Find $g$ using \eqref{eq:Rescale the equalizer}. The new
equalizer is $\check{\mathbf{u}}=g\mathbf{u}$, and the new $\tilde{\mathbf{u}}^{\left(k+1\right)}=\left[\left(\text{qvec}\left(\check{\mathbf{u}}\right)\right)^{T}\;1\right]^{T}$. 
\item Step 5. Repeat step 2 until converge. 
\end{itemize}

\subsection{Complexity discussion}
The complexity of the system depends mostly on the size of $\mathbf{G}$. It affects the size of kurtosis and covariance matrices, the number of variables input to the SDP solver, the size of vectors and matrices in post-processing step. 
Similar to many other batch algorithms the complexity of forming the statistical matrices is $\mathcal{O}(KN^4)$.
The SDP solver using interior point method requires the worst case complexity of $O(n^{3.5})$, where $n$ is the number of variables. 
As the size of matrix $\mathbf{G}$ is $\mathcal{O}(N^2)\mathcal{O}(N^2)$, the SDP-solver requires $\mathcal{O}((N^4)^{3.5})$ arithmetic operations.
For the post-processing techniques, the dominant operation is the rank one approximation which is realized by singular value decomposition and its complexity is $\mathcal{O}(N^3)$.
We further discuss about the SDP solver complexity as it is dominant.
Compared to the work in \cite{Maricic2003}, we can see that, in big O notation, the worst case complexity do not change. However, in practice, convex optimization in \cite{Maricic2003} estimates $2N$ complex equalizer taps  whereas the proposed CO-CMA method estimates $2N$ many real equalization taps resulting into  reduction of worse case complexity. That is equivalent to reduce the complexity of the SDP solver by $2^{3.5}$ for the best case. This, in practice, is very significant.

In \cite{Maricic2003}, the authors mentioned about sparsity of $\mathbf{G}$ for the formulation of CMA cost where the first and the third order parts are zeros and a specially tailored SDP solver is needed. With our formulation, we show that it is not necessary to have such a solver. By omitting the odd order parts, the new formulation further reduces the number of variables from $(2N^2+3N+1)\times(2N^2+3N+1)$ to $(2N^2+N+1)\times(2N^2+N+1)$. For $N=6$, it is equivalent to reducing $37\%$ complexity.

Another advantage of our formulation over the formulation in \cite{Maricic2003} is that the vector $\mathbf{u}$ comprises the real and imaginary parts of the equalizer whereas the corresponding component in \cite{Maricic2003} is made of the equalizer vector and its Hermitian. In \cite{Maricic2003}, this dependency is not imposed in the post-processing step and may requires more iterations to converge. In fact, our simulations show that at most 3 iterations are requires to converge where as the work in \cite{Maricic2003} requires 5 iterations. 

As compared to the traditional gradient decent algorithms, the convex optimization approach, which uses SDP solver for the fourth order cost, provides much better performance at the expanse of complexity. The gradient decent algorithms such as BGD-CMA or OS-CMA, have low complexity and  suffer from local minima and, in many cases, provide poor performance. In contrast, the convex optimization approach at least finds a local minimum that achieves good performance. Due to large computational complexity, the convex optimization approach is not viable for real time applications. Nevertheless, our work achieves incremental complexity reduction for CO-CMA method.

\section{Generalization to Other Algorithms}

\label{sec:CO-Generalize SWA MED Training} The principle of the convex
optimization relaxation and the accompanied iterative mapping procedure
can be generalized beyond CMA. In fact, several well-known blind algorithms
based on fourth order statistics can also be recast into convex optimization
algorithms.

\subsection{Shalvi-Weinstein algorithm}

Closely related to CMA is SWA \cite{Shalvi1990}. This algorithm is
also based on the fourth order statistics by minimizing the following
cost 
\begin{align}
\begin{split}
J_{\text{SWA}}  =&E\left[\left|y\left(k\right)\right|^{4}\right]-\left(2+\frac{\left(1+\alpha\right)\gamma_{s}}{\sigma_{s}^{4}}\right)\\
&E^{2}\left[\left|y\left(k\right)\right|^{2}\right] +2\alpha\frac{\gamma_{s}}{\sigma_{s}^{2}}E\left[\left|y\left(k\right)\right|^{2}\right].\label{eq:Shalvi-Weinstein Algorithm formulation}
\end{split}
\end{align}
Here, the source index $i$ is omitted. For $\alpha=-R_{2}\sigma_{s}^{2}/\gamma_{s}$,
the SWA cost becomes the CMA cost with a constant difference.

Using real-valued vector notation, we can rewrite the SWA cost as
\begin{align}
J_{\text{SWA}}  =&\mathbf{v}^{T}\mathbf{C}\mathbf{v}-\left(2+\frac{\left(1+\alpha\right)\gamma_{s}}{\sigma_{s}^{4}}\right)\mathbf{v}^{T}\mathbf{b}\mathbf{b}^{T}\mathbf{v}+2\alpha\frac{\gamma_{s}}{\sigma_{s}^{2}}\mathbf{v}^{T}\mathbf{b}\nonumber\\
  =&\mathbf{v}^{T}\left[\mathbf{C}-\left(2+\frac{\left(1+\alpha\right)\gamma_{s}}{\sigma_{s}^{4}}\right)\mathbf{b}\mathbf{b}^{T}\right]\mathbf{v}+\\
&2\left(\alpha\frac{\gamma_{s}}{\sigma_{s}^{2}}\mathbf{b}^{T}\right)\mathbf{v}.\nonumber
\label{eq:Shalvi-Weinstein Algorithm formulation real}
\end{align}
Clearly, the cost function is a fourth order function of the equalizer
parameter vector $\mathbf{u}$, similar to the CMA. Hence, we can
apply the proposed convex optimization method on this cost function
for convergence enhancement.

\subsection{MED cost and the modified cost for convex optimization}

Another related algorithm is the MED algorithm originally developed
by Donoho in \cite{Donoho1981} as 
\begin{eqnarray}
 & \text{Maximize} & \text{sign}\left(\gamma_{s}\right)E\left[\left|y\left(k\right)\right|^{4}\right]\label{eq:Minimum Entropy Deconvolution}\\
 & \text{subject to} & E\left[\left|y\left(k\right)\right|^{2}\right]=\sigma_{s}^{2}.\nonumber 
\end{eqnarray}
Taking into account the negative kurtosis of the communication signals,
i.e., $\text{sign}\left(\gamma_{s}\right)=-1$, we form an equivalent
cost 
\begin{equation}
J_{\text{med}}=E\left[\left|y\left(k\right)\right|^{4}\right]+\lambda_{\text{p}}\left(E\left[\left|y\left(k\right)\right|^{2}\right]-\sigma_{s}^{2}\right)^{2}
\end{equation}
where $\lambda_{\text{p}}$ is the Lagrangian multiplier for the power
constraints. The cost can be further summarized as 
\begin{align}
\begin{split}
J_{\text{med}}  =&E\left[\left|y\left(k\right)\right|^{4}\right]+\lambda_{\text{p}}E^{2}\left[\left|y\left(k\right)\right|^{2}\right]-\\
&\lambda_{\text{p}}2\sigma_{s}^{2}E\left[\left|y\left(k\right)\right|^{2}\right]+\lambda_{\text{p}}\sigma_{s}^{4}.
\end{split}
\end{align}
Applying the real-valued formulation, we arrive at another fourth order
blind cost function 
\begin{align*}
J_{\text{med}} & =\mathbf{v}^{T}\left[\mathbf{C}+\lambda_{\text{p}}\mathbf{b}\mathbf{b}^{T}\right]\mathbf{v}-2\left(\lambda_{\text{p}}\sigma_{s}^{2}\mathbf{b}^{T}\right)\mathbf{v}+\lambda_{\text{p}}\sigma_{s}^{4}.
\end{align*}
Therefore, we can apply the proposed convex optimization method on
MED for convergence enhancement.

\section{Application to Semiblind Equalization and Signal Recovery}

In practice, there are often pilot symbols for channel estimation
and equalization. Semiblind equalization and signal recovery are desirable
when the linear system is too complex for the available pilot symbols
to fully estimate or equalize. In such cases, integrating pilot symbols
with blind recovery criterion into a semiblind equalization makes
better sense. In this section, we present a semiblind algorithm by
utilizing the available pilot samples to enhance the equalization
performance using the aforementioned convex optimization formulation.

Our basic principle is to construct a special semiblind cost for minimization.
Typically, such cost functions are usually formed as a linear combination
of a blind criterion $J_{\text{b}}$ and a pilot based cost $J_{\text{t}}$
in the form of 
\begin{align}
J_{\text{sb}} & =\lambda J_{\text{b}}+(1-\lambda)J_{\text{t}}.\label{eq:Semiblind cost - general form}
\end{align}
Here the scalar $\lambda$ depends on a number of factors such as
the pilot sequence length, the number of sources, as well as the source
signal constellations.

In order to apply the same convex optimization principle, we aim to
develop a semiblind cost function that is also a fourth order function
of the receiver parameters. This cost function $J_{\text{sb}}$ should
also only have even order coefficients. Here, for simplicity, we discuss
the use of pilot in single source recovery mode (SISO or SIMO channel
model), although our method can be easily generalized for multiple
sources.

For a single signal source, we can omit the source index. Without
loss of generality, let the transmitted symbol $s(k)$, $k=1,...,L_{\text{t}}$
be the training sequence. In the absence of noise and under perfect
channel equalization, we have the following linear relationship between
training symbols and received signals 
\begin{equation}
\mathbf{w}^{H}\mathbf{x}(k)=s(k-d),\qquad\; n=1,...,L_{\text{t}}.\label{eq:Training Equation, Complex}
\end{equation}
where $d$ is the decision delay.

We convert the complex data representation to real-valued data representation
by defining 
\begin{equation}
s_{\text{t}}(p)=\begin{cases}
\mbox{Re}\left\{ s(\frac{p}{2}-d)\right\}  & \mbox{ if \ensuremath{p} is even}\\
\mbox{Im}\left\{ s(\frac{p+1}{2}-d)\right\}  & \mbox{ if \ensuremath{p} is odd}
\end{cases}
\end{equation}
and 
\begin{equation}
\mathbf{x}_{\text{t}}(p)=\begin{cases}
\mathbf{x}_{\text{r}}(\frac{p}{2}) & \mbox{ if \ensuremath{p} is even}\\
\mathbf{x}_{\text{i}}(\frac{p+1}{2}) & \mbox{ if \ensuremath{p} is odd.}
\end{cases}
\end{equation}
The relationship between training vector and the received signal vector
can be equivalently written as real-valued signal representation as follows
\begin{equation}
\mathbf{u}^{T}\mathbf{x}_{\text{t}}(p)=s_{\text{t}}(p),\qquad p=1,\;2,\;\cdots,\;2L_{\text{t}}.\label{eq:Training Equation, Real}
\end{equation}
This relationship can be also written as a fourth order function of
$\mathbf{u}$ through the following steps.

First, under noise free condition, if the channel can be completely
equalized, then the following criteria can be used for channel equalization
\begin{equation}
\mathbf{u}^{T}\mathbf{x}_{\text{t}}(p)\mathbf{x}_{\text{t}}^{T}(q)\mathbf{u}=s_{\text{t}}(p)s_{\text{t}}(q),\qquad p,q=1,\;2,\;\cdots,\;2L_{\text{t}}.\label{eq:Training Equation, Real-2}
\end{equation}
Therefore, the equalizer should minimize the following pilot-based
fourth order cost function 
\begin{align}
J_{\text{t}} & =\sum_{p=1}^{2L_{\text{t}}}\sum_{q=p}^{2L_{\text{t}}}\left[\mathbf{u}^{T}\mathbf{x}_{\text{t}}(p)\mathbf{x}_{\text{t}}^{T}(q)\mathbf{u}-s_{\text{t}}(p)s_{\text{t}}(q)\right]^{2}\nonumber \\
 & =\mathbf{v}^{T}\mathbf{C}_{\text{t}}\mathbf{v}-\mathbf{b}_{\text{t}}\mathbf{v}+a_{\text{t}},
\end{align}
where $\mathbf{C}_{\text{t}}$, $\mathbf{b}_{\text{t}}$, and $a_{\text{t}}$
are appropriately defined matrix, vector, and scalar that correspond
to the fourth, the second, and the zero-th order coefficients of $J_{\text{t}}$,
respectively. Their details are presented in Appendix B.

This cost can be combined with one of the blind costs to form semi-blind
costs which can be solved by the convex optimization 
\begin{equation}
J_{sb}=\mathbf{v}^{T}\mathbf{C}_{\text{sb}}\mathbf{v}-\mathbf{b}_{\text{sb}}\mathbf{v}+a_{\text{sb}}
\end{equation}
where 
\begin{align}
\mathbf{C}_{\text{sb}} & =\lambda\mathbf{C}_{\text{b}}+(1-\lambda)\mathbf{C}_{\text{t}},\\
\mathbf{b}_{\text{sb}} & =\lambda\mathbf{b}_{\text{b}}+(1-\lambda)\mathbf{b}_{\text{t}},\\
a_{\text{sb}} & =\lambda a_{\text{b}}+(1-\lambda)a_{\text{t}}
\end{align}
with $\mathbf{C}_{\text{b}}$, $\mathbf{b}_{\text{b}}$, and $a_{\text{b}}$
are the matrix, vector, scalar that correspond to fourth, second,
zero order coefficients of the aforementioned blind costs, respectively.

\section{Simulation results\label{sec:CO-Simulation results}}

We now present numerical results to illustrate the performance of
our new batch algorithms using convex optimization (CO) for CMA, SWA,
MED and the fourth order cost semiblind algorithm. In batch algorithms,
the expectations are replaced by the average of $K$ equalizer input
data samples for each antenna. To show the advantages of our formulation,
we compare our algorithms with other blind/semiblind algorithms using
gradient descent method. For the best performance, batch gradient
descent algorithms are used and marked as BGD.

The semiblind algorithms are marked as SB. The semiblind cost for
comparison is selected with CMA cost for the blind part and LS cost
for pilots in \eqref{eq:Semiblind cost - general form}: 
\begin{equation}
J_{\text{b}}=J_{\text{cma}},\qquad\qquad J_{\text{t}}=\sum_{k=\delta+1}^{\delta+L_{\text{t}}}\left|\mathbf{w}^{H}\mathbf{x}\left(k\right)-s\left(k-\delta\right)\right|^{2}.\label{eq:ISI performance metric}
\end{equation}
Here, the decision delays $\delta$ is chosen optimally.

We compare the algorithm performance in terms of the final receiver
ISI. Quantitatively, the equalization effectiveness is measured by
the normalized ISI for the $i-$th equalizer output defined as 
\begin{equation}
\text{ISI}_{i}=\frac{\sum_{j,k}\left|c_{i,j}\left(k\right)\right|^{2}-{\max_{j,k}}\left|c_{i,j}\left(k\right)\right|^{2}}{{\max_{j,k}}\left|c_{i,j}\left(k\right)\right|^{2}}.
\end{equation}
For multiple source signals, we also define the sum normalized ISI
as 
\begin{equation}
\text{ISI}=\mbox{\ensuremath{\sum}}_{i}\text{ISI}_{i}.
\end{equation}

There are several systems in our studies: SISO system with multipath,
$2\times2$ MIMO system with multipath and a $4\times4$ mixing matrix.
The convex optimization is proceeded using available optimization
tool \cite{Tutuncu2003}. The results are averaged over 500 Monte-Carlo
realizations of inputs and channels, unless otherwise specified.

\subsection{SISO channel}

We first test our convex optimization using CMA, SWA, MED costs on
SISO channel and compare with BGD-CMA. The QPSK and 16-QAM inputs
are passed through 3-tap Rayleigh channels with unit power profile
and the receive signals are equalized by $6$-tap equalizers. 
Tables \ref{tab:ISI and Cost comparison among CMA, SWA, MED - SISO, QPSK}
and \ref{tab:ISI and Cost comparison among CMA, SWA, MED - SISO, 16QAM}
compare the ISI performance of different algorithms for QPSK and 16-QAM channel inputs, respectively. 
For convex optimization, SDP-3 with interior point method is used. 
For post-processing, the threshold $\gamma$ is set at $10^{-7}$. The CO-CMA uses both post-processing techniques in \cite{Maricic2003} (marked as "pp1") and the newly proposed one (marked as "pp2"). 
Other CO algorithms apply only the proposed post-processing technique.  
The data length is $1000$ to ensure good convergence.
\begin{table}[H]
{\small{{{}{}\protect\protect\caption{{\small{{{}{}The ISI and cost of different algorithms for SISO
system with QPSK inputs\label{tab:ISI and Cost comparison among CMA, SWA, MED - SISO, QPSK}}}}}
}}{\small \par}

{\small{}}}{\small \par}

{\small{{{}{}\medskip{}
 }}}{\small \par}

\centering{}%
\begin{tabular}{|c|c|c|}
\hline 
{\small{}ISI (dB)} & \multicolumn{1}{c|}{No noise (SNR=$\infty$)} & \multicolumn{1}{c|}{SNR=$14$dB}\tabularnewline
\hline 
Optimum  & $-10.1813$  & - \tabularnewline
\hline 
{\small{}CO-CMA (pp1)} & $-10.1102$  & {\small{}$-9.5104$} \tabularnewline
\hline 
{\small{}CO-CMA (pp2)} & $-10.1104$ & {\small{}$-9.5918$} \tabularnewline
\hline 
{\small{}CO-SWA} ($\alpha=0.5$)  & $-10.1106$  & {\small{}$-9.5953$}\tabularnewline
\hline 
{\small{}CO-MED} ($\lambda_{\text{p}}=2$)  & $-10.1060$  & {\small{}$-9.5892$}\tabularnewline
\hline 
{\small{}BGD-CMA-(1)} & $-8.0446$  & {\small{}$-7.9024$}\tabularnewline
\hline 
{\small{}BGD-CMA-(3)} & $-8.6376$  & {\small{}$-8.4791$}\tabularnewline
\hline 
\end{tabular}
\end{table}

\begin{table}[H]
{\small{{{}{}\protect\protect\caption{{\small{{{}{}The ISI and cost of different algorithms for SISO
system with 16-QAM inputs\label{tab:ISI and Cost comparison among CMA, SWA, MED - SISO, 16QAM}}}}}
}}{\small \par}

{\small{}}}{\small \par}

{\small{{{}{}\medskip{}
 }}}{\small \par}

\centering{}%
\begin{tabular}{|c|c|c|}
\hline 
{\small{}ISI (dB)} & \multicolumn{1}{c|}{No noise (SNR=$\infty$)} & \multicolumn{1}{c|}{SNR=$14$dB}\tabularnewline
\hline 
Optimum  & $-10.1813$  & - \tabularnewline
\hline 
{\small{}CO-CMA (pp1)} & $-9.7162$  & {\small{}$-9.2021$} \tabularnewline
\hline 
{\small{}CO-CMA (pp2)} & $-9.7169$  & {\small{}$-9.2186$}\tabularnewline
\hline 
{\small{}CO-SWA} ($\alpha=5$)  & $-9.7166$  & {\small{}$-9.2196$}\tabularnewline
\hline 
{\small{}CO-MED} ($\lambda_{\text{p}}=2$)  & $-9.7170$  & {\small{}$-9.2161$}\tabularnewline
\hline 
{\small{}BGD-CMA-(1)}   & $-7.9992$  & {\small{}$-6.4451$}\tabularnewline
\hline 
{\small{}BGD-CMA-(3)}  & $-8.3059$  & {\small{}$-7.2917$}\tabularnewline
\hline 
\end{tabular}
\end{table}

The simulations show no significant performance difference among the algorithms using convex optimization. 
Under good condition, i.e., high SNR and long enough data size, the convex optimization
could find the global minima for each case. In fact, the equalization
objective does not depend on the costs. The difference among the costs
with the choice of parameters are perhaps the convergence of the gradient
descent implementations as confirmed by many works \cite{Chen2004,Han2010,Ding2001,Chen1991}.
Similar results can be observed in MIMO channel. Therefore, in the
rest of the simulation results, we only consider CMA as an example
of fourth order statistic based blind channel equalization cost. We
also compare the CO algorithms with BGD-CMA in Tables \ref{tab:ISI and Cost comparison among CMA, SWA, MED - SISO, QPSK}
and \ref{tab:ISI and Cost comparison among CMA, SWA, MED - SISO, 16QAM},
The BGD-CMA-(1) and BGD-CMA-(3) are initialized with single spike
at the first and the third position, respectively. We can see that the initial
value is essential for BGD algorithms as they converge to local minima.
On the other hand, the CO algorithms do not depend on the initial
value to reach a near global optima.

Comparing two post-processing techniques for CO-CMA for floating point implementation, our newly proposed method is slightly better than the existing method in \cite{Maricic2003}. We observed the cases when the last element of $\hat{\mathbf{u}}_{\text{opt}}$ is nearly zero. In fact, the probability that the last element of $\hat{\mathbf{u}}_{\text{opt}}$ being zeros is very small for floating point implementation resulting in to similar performance. However, for fixed point implementation, the probability of last element being zero increases by decreasing the number of fractional bits, which results in to divergence of equalizer of pp1 due to normalization. 
In order to investigate impact of fractional bits of fixed point implementation on \cite{Maricic2003} and proposed method, we consider 8, 9, 10, 11 and 12 fractional bits for both pp1 and pp2 with 6-tap equalizer. In simulation setup, we use 16-QAM constellation and complex channel from \cite{Dogancay1999} as 
\begin{equation}
\mathbf{h}=\left[\begin{array}{c}
-0.033+0.014j\\
0.085-0.039j\\
-0.232+0.136j\\
0.634-0.445j\\
0.07-0.233j\\
-0.027-0.071j\\
-0.023-0.012j
\end{array}\right].
\label{eq:Dogancay channel}
\end{equation}

\begin{figure}[H]
\begin{centering}
\includegraphics[scale=0.8]{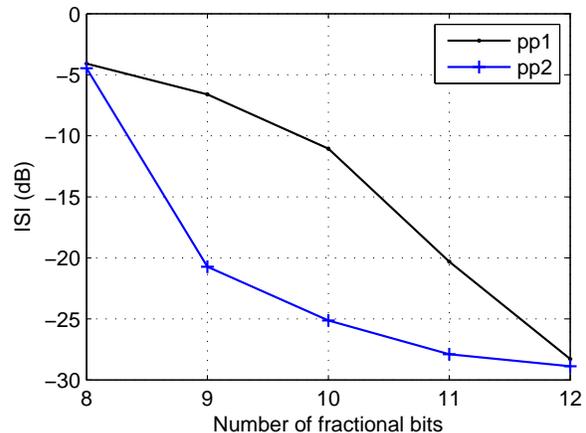} 
\par\end{centering}

\protect\protect\caption{ISI vs number of fractional bits comparison between two post-processing techniques for CO-CMA for SISO channel, 16 QAM signal (SNR=$18$dB, $K=1000$, $L_{\text{w}}=5$)}

\label{fig:ISI vs Qx CO-CMA NM1 16QAM} 
\end{figure}

Figure \ref{fig:ISI vs Qx CO-CMA NM1 16QAM} compares the residual ISI of the two aforementioned methods averaged over $500$ Monte Carlo runs. As figure \ref{fig:ISI vs Qx CO-CMA NM1 16QAM} reveals, pp2 outperforms pp1 when the number of fractional bits are from 9 to 11. 
It is clear that when the number of fractional bits is not sufficient, the small values become zeros and the pp1 suffers from division by zero problem. In contrast, the pp2 does not have the normalization step as in pp1 making pp2 more robust for fixed point implementation.


\begin{figure}[t]
\begin{centering}
\includegraphics[scale=0.60]{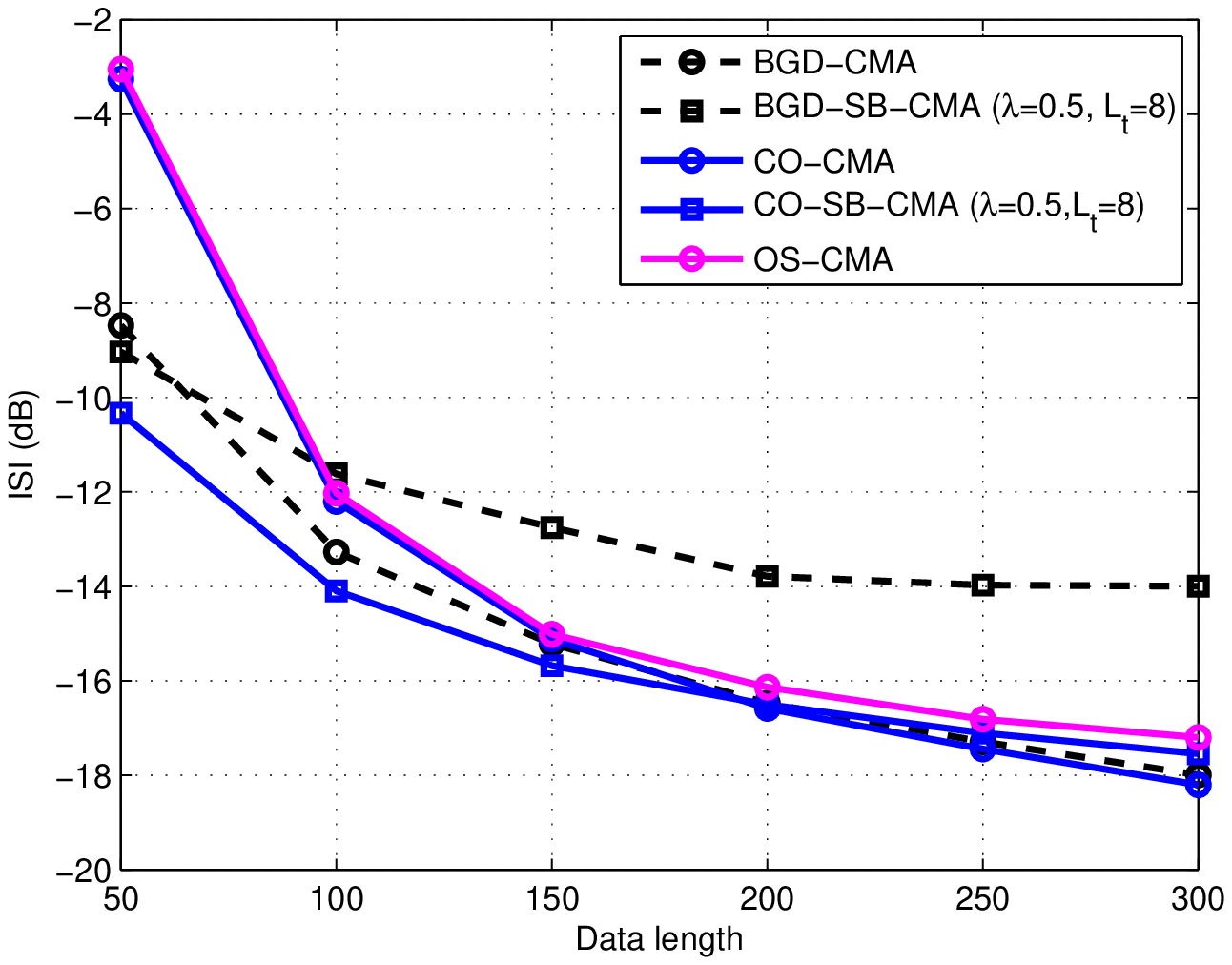} 
\par\end{centering}

\protect\protect\caption{ISI performance of different blind/semiblind algorithms using convex
optimization for the SISO channel with $4$-QAM input (SNR=$8$dB).}

\label{fig:ISI vs K SB-CO-CMA SISO 4QAM} 
\end{figure}

\begin{figure}[t]
\begin{centering}
\includegraphics[scale=0.60]{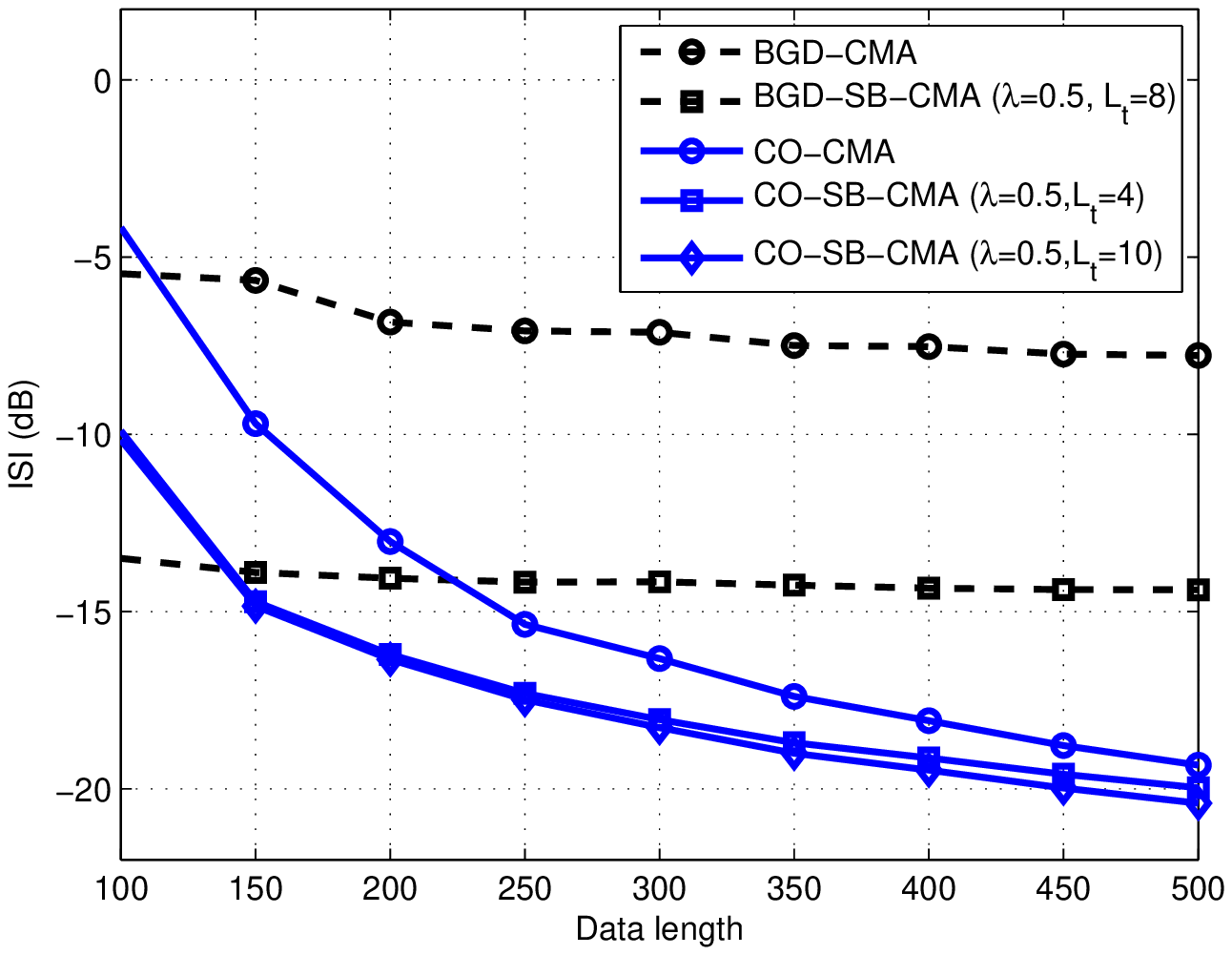} 
\par\end{centering}

\protect\protect\caption{ISI performance of different blind/semiblind algorithms using convex
optimization for the SISO channel with $16$ QAM input (SNR=$14$dB).}

\label{fig:ISI vs K SB-CO-CMA SISO 16QAM} 
\end{figure}

To assert the behavior of the convex optimization algorithms, we test
the algorithms on the fix channel in \eqref{eq:Dogancay channel}. In the rest of simulation results, for convex optimization, only the proposed post-processing technique with floating point is used. The equalizer length is $6$. Figures \ref{fig:ISI vs K SB-CO-CMA SISO 4QAM}
and \ref{fig:ISI vs K SB-CO-CMA SISO 16QAM} show the ISI performance
of the semiblind algorithms for the SISO channel with QPSK and 16-QAM
signals, respectively, when data length $K$ varies. The SNR in QPSK system is 8 dB and in 16-QAM system is 14 dB, For semiblind algorithms, $\lambda=0.5$ and $L_{\text{t}}=8$.
For BGD algorithms, the initialization is $[0\;0\;1\;0\;0\;0]$. For
good ISI performance, the step size is set at $0.01$ and $0.001$
for QPSK and 16-QAM, respectively. With QPSK inputs, the blind algorithms
have almost the same performance for a large enough $K$. It is known
that the channel is fairly easy to equalize and with the initialization,
the BGD actually converges to the global minima. The performance of
semiblind algorithms are worse than that of the corresponding blind
ones. It is expectable since the pilots may deviate the semiblind
costs from a good result due to the instantaneous noise value on the
pilots \cite{Zarzoso2005a}. With 16-QAM inputs, the blind cost appears
to have many local minima and the BGD-CMA cannot converge to good
local minima even when the data length is considered sufficient.
The pilots help the BGD-SB-CMA to eliminate many unsatisfactory local
minima. However, it still cannot achieve global minimum. On the
other hand, the CO-CMA and the CO-SB-CMA have substantially better performance
than the BGD counterpart. Comparing the CO-CMA and the CO-SB-CMA for large $K$, i.e., $K$=500, we can see that the performance difference is negligible even when we increase the pilot length.
This indicates that the CO-CMA converge to global minimum. For smaller
$K$, as $K$ decreases, the performance of all the convex optimization based algorithms degrade since the estimation
of the statistics become less accurate. With only a few pilot symbols, the
SB-CO-CMA reduces ISI substantially, especially when the data length is small.

\begin{figure}[t]
\begin{centering}
\includegraphics[scale=0.60]{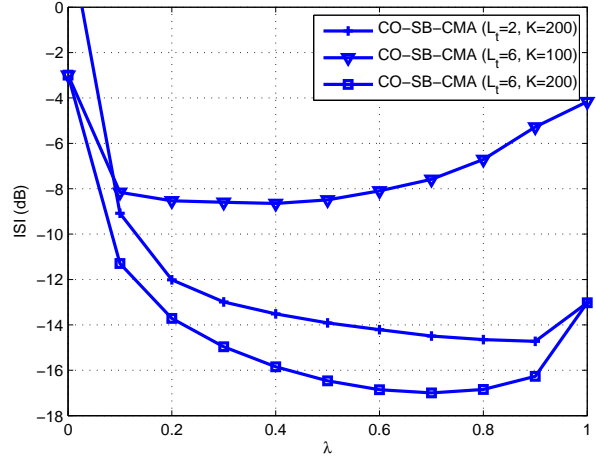} 
\par\end{centering}

\protect\protect\caption{ISI vs weighting factor $\lambda$ for CO-SB-CMA for SISO channel
and 16-QAM input (SNR=$14$dB).}

\label{fig:ISI vs weight SB-CO-CMA SISO 16QAM} 
\end{figure}

\begin{figure}[t]
\begin{centering}
\includegraphics[scale=0.60]{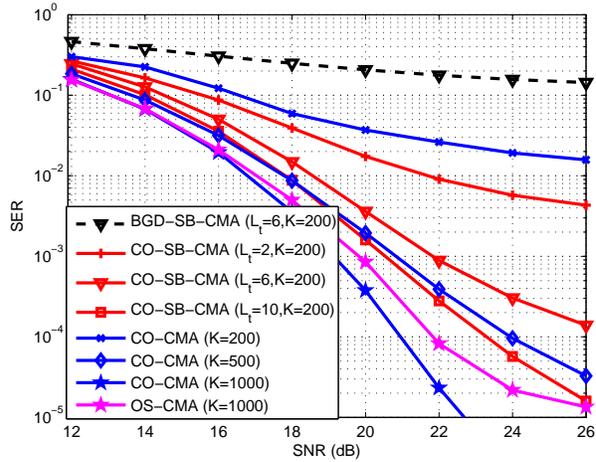} 
\par\end{centering}

\protect\protect\caption{SER vs SNR comparison between CMA and the SB-CMA using convex optimization
(SISO channel, 16 QAM signal).}

\label{fig:SER vs SNR SB-CO-CMA SISO 16QAM} 
\end{figure}

We also compare the CO-CMA and the optimal stepsize CMA algorithm in \cite{Zarzoso2005a} and \cite{Zarzoso2008} (marked as "OS-CMA"). For the chosen simulation parameter set, the CO-CMA outperforms the OS-CMA slightly.

Figure \ref{fig:ISI vs weight SB-CO-CMA SISO 16QAM} shows the ISI
versus $\lambda$ of the SB-CO-CMA for several setups. The semiblind
algorithm always outperforms the blind counterpart and the training
based one. The result also indicates that the weighting factor must
be carefully chosen . The optimum weight depends on the choice of
the length of training symbols and the length of available data.

Figure \ref{fig:SER vs SNR SB-CO-CMA SISO 16QAM} shows the symbol
error rate (SER) versus SNR performance for the blind and semiblind
algorithms using convex optimization under 16-QAM signal. We can see
that the CO algorithms improve the performance significantly when
comparing with the BGD algorithm. Figure \ref{fig:SER vs SNR SB-CO-CMA SISO 16QAM}
also confirms that with the help of a few pilot symbols, the semiblind
algorithms significantly reduce the need for long data sequence. In
particular, with only $6$ or $10$ pilot symbols, the data length
$K$ for the SB algorithm is only $200$ compare to $500$ for the
blind algorithm. Therefore, the CO processing technique is suitable
for short burst data transmission. 
Compared to OS-CMA, the performance of CO-CMA is significant better. This is because the cost for 16-QAM appears to have many local minima and OS-CMA does not guarantee to converge to the global one.

\subsection{Blind source separation problem}

In this part, we investigate the source separation ability of the
new algorithms. We test the algorithms on a $4\times4$ mixing matrix
for source separation problem. In this case, by limiting $k=0$ in
the performance metric in \eqref{eq:ISI performance metric}, we define
the measure metric as the calculated normalized cross-channel interference
(NCCI) for source $i$. 
\begin{equation}
\text{NCCI}_{i}=\frac{\sum_{j}\left|c_{i,j}\right|^{2}-\text{max}_{j}\left|c_{i,j}\right|^{2}}{\text{max}_{j}\left|c_{i,j}\right|^{2}}.
\end{equation}
In the following simulation results, we use the average NCCI of all
the combined channels to compare different simulation setups.

In our example, the mixing matrix $\mathbf{H}$ is $4\times4$ with
entries 
\begin{align*}
\mathbf{H} &\!\! =\!\!\!\left[\begin{array}{cccc}
\!\!\!0.41+0.05j & \!\!\!\!\!\!0.45+0.62j & \!\!\!0.26+0.92j & \!\!-0.25-0.61j\\
\!\!\!0.52-1.11j & \!\!\!\!\!1.04-0.12j & \!\!\!0.06+0.66j & \!\!-0.81+0.21j\\
\!\!\!0.07-0.80j &\!\!\!\! 1.30+0.33j & \!\!\!1.40+0.65j &\!\! -0.05+0.94j\\
\!\!\!0.47-1.08j &\! \!\!0.83+0.43j & \!\!\!0.94-0.08j & \!\!0.57+0.19j
\end{array}\!\!\!\!\right]
\end{align*}

Each of the 4 channel outputs has AWG noise with SNR of $10$dB. We
apply the CMA algorithms without cross-cumulant jointly with the Gram-Schmidt
orthogonalization process to separate different sources \cite{Ding2000a}.
For BGD algorithms, the initial values for $\{\mathbf{w}_{i},\; i=1,2,3,4\}$
are $\left[1\:0\:0\:0\right]^{T}$, $\left[0\:1\:0\:0\right]^{T}$,
$\left[0\:0\:1\:0\right]^{T}$, $\left[0\:0\:0\:1\right]^{T}$. The
step size is set at $10^{-3}$ for $4$-QAM and $10^{-5}$ for $16$-QAM
signals and the update iterations are long enough for the best NCCI
performance. For the semiblind algorithm, the pilot for each data
stream is chosen orthogonal. Although the training symbols can do
the source separation task, we still use the cross correlation cost
as the semiblind cost for the second channel given by $J_{\text{sb},2}=(1-\lambda)J_{\text{t},2}+\lambda J_{\text{b},2}+\lambda_{\text{cr}}J_{\text{cr},2,1}$
with $\lambda_{\text{cr}}=1$. If $\lambda=1$, the semiblind algorithms
reduce to the blind algorithms. This semiblind cost is also used in
the simulation for MIMO channel.

\begin{figure}[t]
\begin{centering}
\includegraphics[scale=0.60]{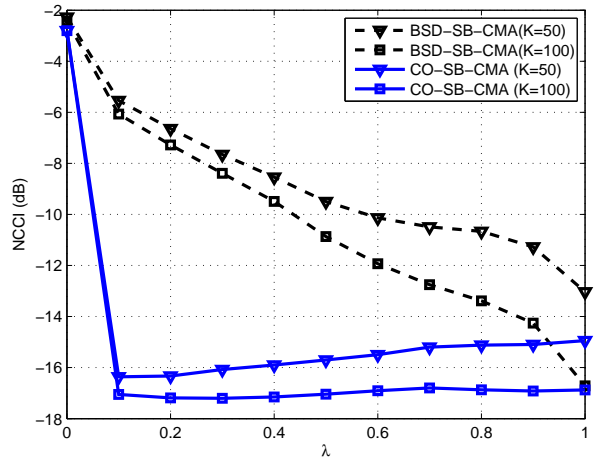} 
\par\end{centering}

\protect\protect\caption{NCCI vs weighting factor for the mixing matrix with QPSK signals ($L_{t}=4)$.}

\label{fig:NCCI vs weight SB-CO-CMA 4x4 4QAM} 
\end{figure}

\begin{figure}[t]
\begin{centering}
\includegraphics[scale=0.60]{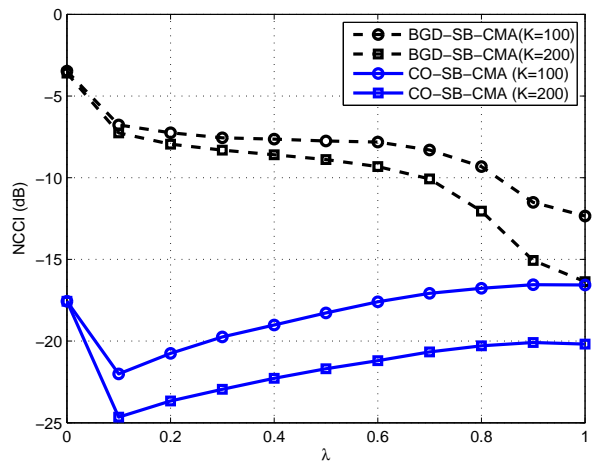} 
\par\end{centering}

\protect\protect\caption{NCCI vs weighting factor for the mixing matrix with 16-QAM signal
($L_{t}=4)$.}

\label{fig:NCCI vs weight SB-CO-CMA 4x4 16QAM } 
\end{figure}
Figures \ref{fig:NCCI vs weight SB-CO-CMA 4x4 4QAM} and \ref{fig:NCCI vs weight SB-CO-CMA 4x4 16QAM }
compare the NCCI between the semiblind algorithms using CMA cost when
the channel inputs are $4$-QAM and $16$-QAM, respectively. The solutions
are achieved by the proposed convex optimization or by the BGD. For
the BGD the pilot cost is the linear least square (LS) cost which
is commonly used. We can see that for short data length, the convex
optimization approach outperforms the BGD one. In this scenario, few
pilots do not help to improve the performance of BGD. In addition,
it creates undesired local minima that the BGD may converge to. Meanwhile,
despite of the bad local minima, the performance of the CO-SB-CMA is superior
to the BGD algorithms. For 16-QAM, the NCCI is lower as the $\lambda$
is decreased. This implies that for few observations ($K=100,200)$
the estimation of high-order statistics is not accurate and the semiblind
algorithms rely on pilots to achieve good performance. For QPSK case,
we observe the same behavior although the performance difference between
the CO-SB-CMA and the CO-CMA is smaller.

Observing figures \ref{fig:NCCI vs weight SB-CO-CMA 4x4 4QAM} and
\ref{fig:NCCI vs weight SB-CO-CMA 4x4 16QAM }, we notice an interesting
phenomenon about the pilot based costs. When $\lambda=1$, the SB-BGD-CMA
cost reduces to the linear least square cost. If the pilot length
is short compared with the equalizer length, as the algorithm converges,
over fitting problem occurs. It is because in the LS cost there are
only $L_{\text{t}}$ square terms. Meanwhile, for the SB-CO-CMA, the fourth
order pilot cost has $L_{\text{t}}(L_{\text{t}}-1)/2$ square terms
making this cost more resilient to the overfitting problem. That is
the reason why even with only a few pilot symbols, the fourth order cost
outperforms the LS cost. We also notice that this performance gap
is more significant for $16$-QAM input. This phenomenon may need
further analysis.

\subsection{MIMO channel}

\begin{table}[t]
{\small{{{}{}\protect\protect\caption{{\small{{{}{}The $2\times2$ MIMO channel coefficients ( $L_{h}=2$)\label{tab:MIMO channel}}}}}
}}{\small \par}

{\small{}}}{\small \par}

{\small{{{}{}\medskip{}
 }}}{\small \par}

\centering{}%
\begin{tabular}{|c|c|c|c|c|}
\hline 
{\small{{{}{}$k$}}}  & {\small{{{}{}$h_{1,1}\left(k\right)$}}}  & {\small{{{}{}$h_{1,2}\left(k\right)$}}}  & {\small{{{}{}$h_{2,1}\left(k\right)$}}}  & {\small{{{}{}$h_{2,2}\left(k\right)$}}}\tabularnewline
\hline 
\hline 
{\small{{{}{}0}}}  & {\small{{{}{}$-0.2+0.1j$}}}  & {\small{{{}{}$.1j$}}}  & $0.1j$  & $1$\tabularnewline
\hline 
{\small{{{}{}1}}}  & {\small{{{}{}$1$}}}  & {\small{{{}{}$0.2$}}}  & $0.1$  & $0.1j$\tabularnewline
\hline 
{\small{{{}{}2}}}  & {\small{{{}{}$0.2j$}}}  & {\small{{{}{}$0.11$}}}  & $0.2j$  & $0.1+0.1j$\tabularnewline
\hline 
\end{tabular}
\end{table}

\begin{figure}[t]
\begin{centering}
\includegraphics[scale=0.60]{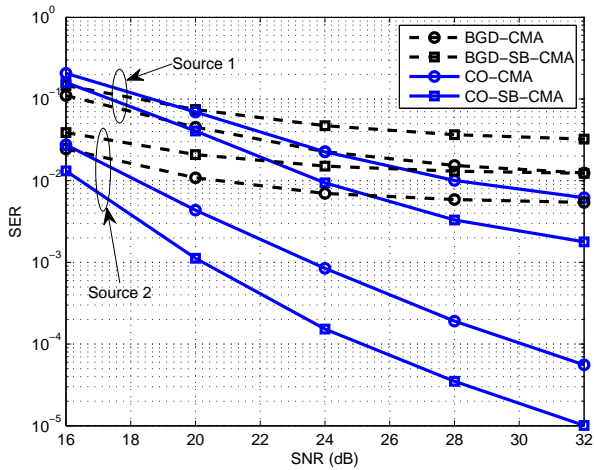} 
\par\end{centering}

\protect\protect\caption{SER vs SNR comparison between CMA and the SB-CMA using convex optimization
for MIMO channel, 16 QAM signal ($K=500$, $L_{\text{t}}=4$)}

\label{fig:SER vs SNR SB-CO-CMA MIMO 16QAM} 
\end{figure}

In this test, we apply our CO approaches using CMA and SB-CMA on a
MIMO channel. We consider a $2\times2$ FIR channel whose coefficients
are shown in Table \ref{tab:MIMO channel}. The signal inputs are
16-QAM and the additive white Gaussian noise is added to form channel
outputs. For each sub-channel, an FIR equalizer of order $L_{w}=3$
is applied. We choose orthogonal pilot sequences of length $L_{\text{t}}=4$
symbols each. The weighting factor $\lambda=0.5$ for SB algorithms
is considered. We compare the CO algorithms with the BGD algorithms. The
stepsize for the BGD is optimized to achieve minimum ISI. The initial
coefficient set is single unit spike at its near-center tap for $\mathbf{w}_{i,i}$
and zeros otherwise: 
\begin{eqnarray}
\mathbf{w}_{i,j} & = & \begin{cases}
[0\;1\;0\;0]^{T} & \text{for }i=j,\\
{}[0\;0\;0\;0]^{T} & \text{for }i\neq j.
\end{cases}
\end{eqnarray}

Figure \ref{fig:SER vs SNR SB-CO-CMA MIMO 16QAM} compares SER performance
of BGD and CO implementations as functions of SNR for CMA and SB-CMA
equalizer for each data source. The results confirm that the CO implementation
also outperforms the BGD in MIMO scenario. Similar to SISO and source
separation scenarios, with only a few pilot symbols, the CO-SB-CMA equalization
achieves better performance over the CO-CMA equalization.

\section{Conclusions\label{sec:CO-Conclusions}}

We formulated various blind channel equalization and source separation costs utilizing only fourth order and second order
statistics into convex semi-definite optimization using real-valued data and parameters. 
Compared to the precedent work in \cite{Maricic2003}, our formulation is more compact and resource saving, thereby, more
efficient. We also proposed a post-processing technique that is more suitable in practice.
Our formulation can be applied for signals of high-order QAM constellation in general MIMO systems. 
The global solution is found without requirement of having good initialization 
as the conventional implementation using gradient descent method.
The number of data required for successfully equalizing the channels or separating sources is low compared to other batch algorithms making this method suitable for short packet transmission under fast fading channels. 
We also proposed a fourth order training based cost and form semi-blind algorithms. 
Using only a few pilot symbols, the semi-blind algorithms outperform the blind cost counterpart and other semi-blind algorithms. 

\begin{appendices} \bibliographystyle{ieeetran}
\bibliography{IEEEabrv,BlindEqualization,math,MyWorks,reference,Semiblind,Standard}

\section{\label{app:CO-Jcr}}

In this appendix, we provide detail calculation for the cross-correlation
term $J_{\text{cr},i,j}$. Note that, $\mathbf{w}_{i}$ for $i<j$
is already known. First, we calculate the cross-correlation as a function
of $\mathbf{v}_{j}=\mbox{qvec}\left(\mathbf{u}_{j}\right)$. Now
\begin{align}
\left|E\left[y_{i}\left(k\right)y_{j}^{*}\left(k-l\right)\right]\right|^{2}&=\left|E\left[\mathbf{w}_{i}^{H}\mathbf{x}\left(k\right)\mathbf{x}^{H}\left(k-l\right)\mathbf{w}_{j}\right]\right|^{2} \nonumber \\
 &=\left|\mathbf{w}_{i}^{H}E\left[\mathbf{x}\left(k\right)\mathbf{x}^{H}\left(k-l\right)\right]\mathbf{w}_{j}\right|^{2}.
\end{align}
Since $\mathbf{w}_{i}$ is known, we can pre-calculate $\mathbf{p}_{i,l}=\left(\mathbf{w}_{i}^{H}E\left[\mathbf{x}\left(k\right)\mathbf{x}^{H}\left(k-l\right)\right]\right)^{H}$.
Let 
\begin{align*}
\mathbf{p}_{\text{r},i,l} & =\left[\begin{array}{cc}
\text{Re}\left\{ \mathbf{p}_{i,l}^{T}\right\}  & \text{Im}\left\{ \mathbf{p}_{i,l}^{T}\right\} \end{array}\right]^{T},\\
\mathbf{p}_{\text{i},i,l}^{T} & =\left[\begin{array}{cc}
\text{Im}\left\{ \mathbf{p}_{i,l}^{T}\right\}  & -\text{Re}\left\{ \mathbf{p}_{i,l}^{T}\right\} \end{array}\right]^{T},
\end{align*}
we have
\begin{align}
\left|E\left[y_{i}\left(k\right)y_{j}^{*}\left(k-l\right)\right]\right|^{2} & =\left|\mathbf{p}_{i,l}^{H}\mathbf{w}_{j}\right|^{2} \nonumber \\
 & =\mathbf{u}_{j}^{T}\left(\mathbf{p}_{\text{r},i,l}\mathbf{p}_{\text{r},i,l}^{T}+\mathbf{p}_{\text{i},i,l}\mathbf{p}_{\text{i},i,l}^{T}\right)\mathbf{u}_{j} \nonumber \\
 & =\mathbf{q}_{i,l}^{T}\mathbf{v}_{j}
\end{align}
with $\mathbf{q}_{i,l}=\mbox{svec}\left(\mathbf{p}_{\text{r},i,l}\mathbf{p}_{\text{r},i,l}^{T}+\mathbf{p}_{\text{i},i,l}\mathbf{p}_{\text{i},i,l}^{T}\right)$.
Therefore 
\begin{align}
\sum_{i=1}^{j-1}J_{\text{cr},i,j} & =\sum_{i=1}^{j-1}\sum_{l=0}^{\delta}\mathbf{q}_{i,l}^{T}\mathbf{v}_{j}\nonumber \\
 & =\mathbf{q}_{j}^{T}\mathbf{v}_{j}
\end{align}
where $\displaystyle \mathbf{q}_{j}=\sum_{i=1}^{j-1}\sum_{l=0}^{\delta}\mathbf{q}_{i,l}.$

\section{\label{app:CO-Jt}}

In this appendix, we present the details of calculation for $\mathbf{C}_{\text{t}}$, $\mathbf{b}_{\text{t}}$, $a_{\text{t}}$
in Section \ref{sec:CO-Generalize SWA MED Training}.

Since 
\begin{align*}
\mathbf{u}^{T}\mathbf{x}_{\text{t}}(p)\mathbf{x}_{\text{t}}^{T}(q)\mathbf{u}&=\mathbf{u}^{T}\mathbf{x}_{\text{t}}(q)\mathbf{x}_{\text{t}}^{T}(p)\mathbf{u}\\
&=\mathbf{u}^{T}\left[\frac{\mathbf{x}_{\text{t}}(p)\mathbf{x}_{\text{t}}^{T}(q)}{2}+\frac{\mathbf{x}_{\text{t}}(q)\mathbf{x}_{\text{t}}^{T}(p)}{2}\right]\mathbf{u},
\end{align*}
we can write

\begin{equation}
J_{\text{t}}=\sum_{p=1}^{2L_{\text{t}}}\sum_{q=p}^{2L_{\text{t}}}\left[\mathcal{X}^{T}(p,q)\mathbf{v}-s_{\text{t}}(p)s_{\text{t}}(q)\right]^{2}
\end{equation}
where $\mathcal{X}(p,q)=\mbox{svec}\left(\frac{\mathbf{x}_{\text{t}}(p)\mathbf{x}_{\text{t}}^{T}(q)}{2}+\frac{\mathbf{x}_{\text{t}}(q)\mathbf{x}_{\text{t}}^{T}(p)}{2}\right)$.

The cost can be written as the second order function of $\mathbf{v}$:

\begin{align}
J_{\text{t}} & =\sum_{p=1}^{2L_{\text{t}}}\sum_{q=p}^{2L_{\text{t}}}\mathbf{v}^{T}\mathcal{X}(p,q)\mathcal{X}^{T}(p,q)\mathbf{v}\\
&-\sum_{p=1}^{2L_{\text{t}}}\sum_{q=p}^{2L_{\text{t}}}2s_{\text{t}}(p)s_{\text{t}}(q)\mathcal{X}^{T}(p,q)\mathbf{v}+\sum_{p=1}^{2L_{\text{t}}}\sum_{q=p}^{2L_{\text{t}}}s_{\text{t}}(p)s_{\text{t}}(q)\nonumber \\
 & =\mathbf{v}^{T}\mathbf{C}_{\text{t}}\mathbf{v}-\mathbf{b}_{\text{t}}\mathbf{v}+a_{\text{t}}
\end{align}

where

\begin{equation}
\mathbf{C}_{\text{t}}=\sum_{p=1}^{2L_{\text{t}}}\sum_{q=p}^{2L_{\text{t}}}\mathcal{X}(p,q)\mathcal{X}^{T}(p,q),
\end{equation}

\begin{equation}
\mathbf{b}_{\text{t}}=\sum_{p=1}^{2L_{\text{t}}}\sum_{q=p}^{2L_{\text{t}}}2s_{\text{t}}(p)s_{\text{t}}(q)\mathcal{X}^{T}(p,q),
\end{equation}

\begin{equation}
a_{\text{t}}=\sum_{p=1}^{2L_{\text{t}}}\sum_{q=p}^{2L_{\text{t}}}s_{\text{t}}(p)s_{\text{t}}(q).
\end{equation}

\end{appendices} 

\vspace*{-2\baselineskip}
\begin{IEEEbiography}
[{\includegraphics[width=1in,height=1.25in,clip,keepaspectratio]{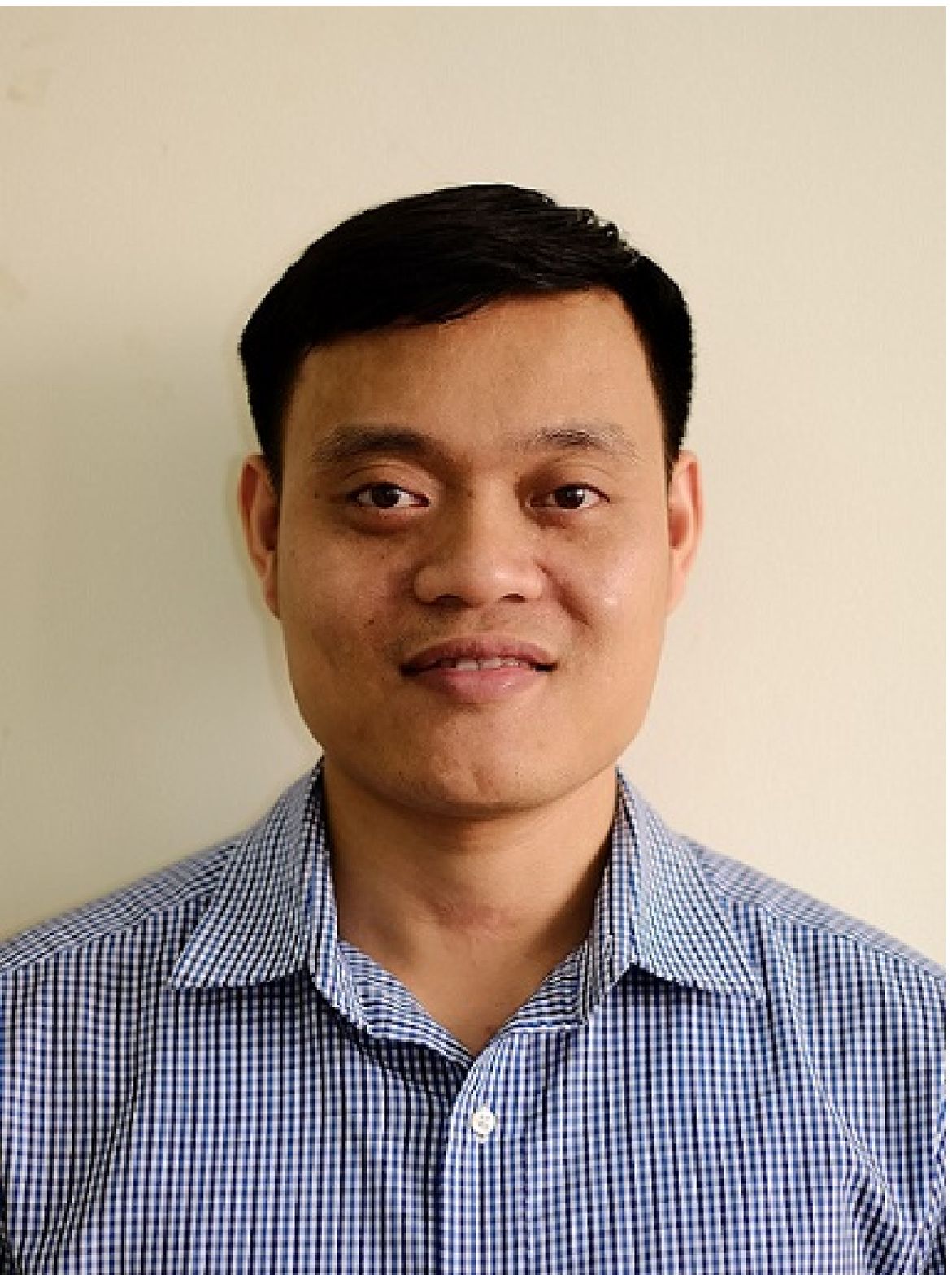}}]{Huy-Dung Han}
Huy-Dung Han received B.S. in 2001 from Faculty of Electronics and Telecommunications at Hanoi University of Science and Technology, Vietnam, M. Sc. degree in 2005 from Technical Faculty at University Kiel, Germany and Ph.D. degree  in 2012, from the Department of Electrical and Computer Engineering at the University of California, Davis, USA. 
His research interests are in the area of wireless communications and signal processing,
with current emphasis on blind and semi-blind channel equalization for single and multi-carrier communication systems, convex optimization.
Dr. Han is with School of Electronics and Telecommunications, Department of Electronics and Computer Engineering at Hanoi University of Science and Technology, Hanoi, Vietnam. 
He has been serving on technical programs of IEEE International Conference on Communications and Electronics.
\end{IEEEbiography}

\vspace*{-2\baselineskip}
\begin{IEEEbiography}
[{\includegraphics[width=1in,height=1.25in,clip,keepaspectratio]{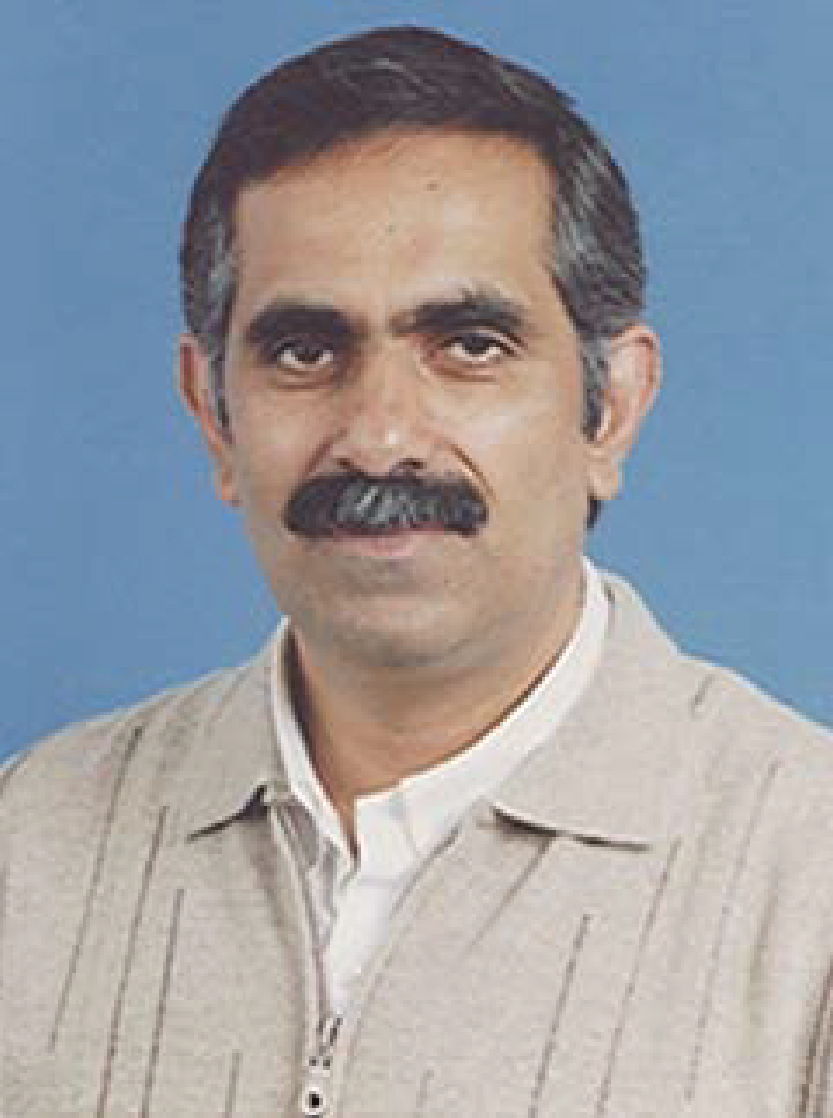}}]{Muhammad Zia}
Muhammad Zia received M. Sc. degree in 1991 and M.Phil degree  in 1999, both from Department of Electronics at Quaid-e-Azam University, Islamabad, Pakistan. He received PhD degree from the Department of Electrical and Computer
Engineering at the University of California, Davis in 2010. His research
interests are in the area of wireless communications and signal processing,
with current emphasis on the wireless security, compressive sensing, bandwidth efficient transceiver design and blind and semi-blind detection of Space-Time Block Codes. Dr. Zia is with Department of Electronics at Quaid-i-Azam University, Islamabad.

\end{IEEEbiography}
\vspace*{-2\baselineskip}
\begin{IEEEbiography}
[{\includegraphics[width=1in,height=1.25in,clip,keepaspectratio]{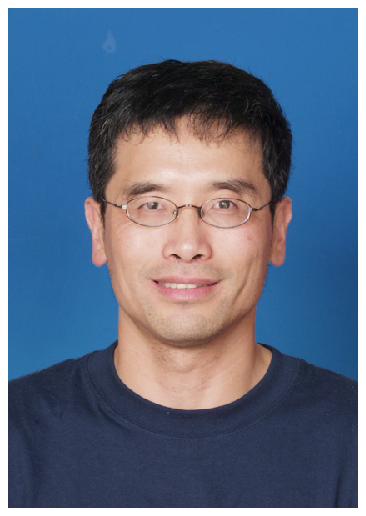}}]{Zhi Ding}
Zhi Ding (S'88-M'90-SM'95-F'03) is Professor of Electrical and Computer Engineering at the University of California, Davis. He received his Ph.D. degree in Electrical Engineering from Cornell University in 1990. From 1990 to 2000, he was a faculty member of Auburn University and later, University of Iowa. Prof. Ding has held visiting positions in Australian National University, Hong Kong University of Science and Technology, NASA Lewis Research Center and USAF Wright Laboratory. Prof. Ding has active collaboration with researchers from areas including Australia, China, Japan, Canada, Finland, Taiwan, Korea, Singapore, and Hong Kong.

Dr. Ding is a Fellow of IEEE and has been an active volunteer, serving on technical programs of several workshops and conferences. He was associate editor for IEEE Transactions on Signal Processing from 1994-1997, 2001-2004, and associate editor of IEEE Signal Processing Letters 2002-2005. He was a member of technical committee on Statistical Signal and Array Processing and member of Technical Committee on Signal Processing for Communications (1994-2003).  Dr. Ding was the Technical Program Chair of the 2006 IEEE Globecom. He was also an IEEE Distinguished Lecturer (Circuits and Systems Society, 2004-06, Communications Society, 2008-09). He served on as IEEE Transactions on Wireless Communications Steering Committee Member (2007-2009) and its Chair (2009-2010). Dr. Ding received the 2012 IEEE Wireless Communication Recognition Award from the IEEE Communications Society and is a coauthor of the text: Modern Digital and Analog Communication Systems, 4th edition, Oxford University Press, 2009. 
\end{IEEEbiography}

\end{document}